\documentclass[11pt]{article}

\usepackage[dvips,letterpaper,margin=1in,bottom=1in]{geometry}

\usepackage{amsmath,amssymb,amsthm}
\usepackage{microtype}
\usepackage{xcolor}
\usepackage{array}
\usepackage{float}
\usepackage{algorithm}
\usepackage{algpseudocode}
\usepackage{authblk}
\algrenewcommand\algorithmicrequire{\textbf{Input:}}
\algrenewcommand\algorithmicensure{\textbf{Output:}}
\usepackage{tikz}
\usetikzlibrary{arrows.meta,calc,positioning}
\usepackage{hyperref}

\floatstyle{ruled}
\restylefloat{algorithm}

\hypersetup{
  colorlinks=true,
  linkcolor=blue!55!black,
  urlcolor=blue!55!black,
  citecolor=blue!55!black,
  pdftitle={Time-Dependent Hamiltonian Simulation with Optimal Query Complexity}
}
\allowdisplaybreaks

\newtheorem{theorem}{Theorem}
\newtheorem{lemma}[theorem]{Lemma}
\newtheorem{proposition}[theorem]{Proposition}
\theoremstyle{definition}
\newtheorem{definition}{Definition}
\newtheorem{convention}{Convention}
\theoremstyle{plain}
\newcommand{\Sys}{\mathcal H_{\mathsf S}}
\newcommand{\Anc}{\mathcal H_{\mathsf A}}
\newcommand{\Time}{\mathcal H_{\mathsf T}}
\newcommand{\Priv}{\mathcal H_{\mathrm{priv}}}
\newcommand{\HAMT}{\mathrm{HAM\mbox{-}T}}
\newcommand{\eps}{\varepsilon}
\newcommand{\norm}[1]{\left\lVert #1\right\rVert}
\newcommand{\ket}[1]{\lvert #1\rangle}
\newcommand{\bra}[1]{\langle #1\rvert}
\newcommand{\ketbra}[2]{\lvert #1\rangle\!\langle #2\rvert}
\newcommand{\bigO}{\mathcal O}

\title{Time-Dependent Hamiltonian Simulation\\with Optimal Query Complexity}

\author[1]{Boyang Chen\thanks{\texttt{by-chen24@mails.tsinghua.edu.cn}}}
\author[2,3]{Minbo Gao\thanks{\texttt{gmb17@tsinghua.org.cn}}}
\author[4,5]{Xinzhao Wang\thanks{\texttt{xinzhaowang3@gmail.com}}}
\author[4,5]{Shuo Zhou\thanks{\texttt{antientropy@pku.edu.cn}}}

\affil[1]{Department of Computer Science and Technology,
Tsinghua University, Beijing, China}

\affil[2]{Institute of Software,
Chinese Academy of Sciences, Beijing, China}

\affil[3]{University of Chinese Academy of Sciences,
Beijing, China}

\affil[4]{Center on Frontiers of Computing Studies,
Peking University, Beijing, China}

\affil[5]{School of Computer Science,
Peking University, Beijing, China}
\date{}

\begin{document}
\maketitle
\vspace{-2.2em}

\begin{abstract}

We give a query-optimal algorithm for simulating a general \(n\)-qubit
time-dependent Hamiltonian \(H(t)\) on \([0,T]\), assuming that \(H\) is
Lipschitz continuous and \(\norm{H(t)}\leq\alpha\).  In the standard
\(\HAMT\) access model, the algorithm approximates the time-ordered
propagator \(U_H(T)\) to error \(\eps\) using
\[
 \bigO\!\left(
 \alpha T+\frac{\log(1/\eps)}
 {\log\!\bigl(e+\log(1/\eps)/(\alpha T)\bigr)}
 \right)
\]
\(\HAMT\) queries.  This matches the known query lower bound for
time-independent Hamiltonians, showing that time dependence incurs no
asymptotic query overhead.

Our method first constructs a one-query transducer that, given an auxiliary
state, implements an approximation to \(U_H(T)\) and returns the state
unchanged.  A weighted combination of circuits that apply the transducer
different numbers of times makes the error caused by omitting this state
decay factorially, yielding the stated optimal
precision dependence.  For time-independent Hamiltonians, the same
method also gives a query-optimal alternative to qubitization.
\end{abstract}

\newpage

\section{Introduction}

Hamiltonian simulation is one of the central algorithmic tasks in quantum
computation. It formalizes the original motivation for quantum computers as
devices for efficiently reproducing quantum dynamics~\cite{Fey82,Llo96}.
In the basic time-independent setting, given access to $H$, the goal is to
implement $e^{-iHt}$ with complexity in terms of the evolution time and target
accuracy. This problem has driven advances from sparse-Hamiltonian algorithms and
linear-combination-of-unitaries methods to query-optimal algorithms based on quantum signal
processing and qubitization, as well as the broader framework of quantum
singular value transformation~\cite{AT03,BACS07,CW12,BCK15,LC17,LC19,GSLW19}.
Hamiltonian simulation also serves as a fundamental subroutine in algorithms
for molecular energy estimation, linear systems of equations, and
thermal-state preparation~\cite{AGDLH05,HHL09,PW09}.

Many natural physical and algorithmic settings instead involve
time-dependent Hamiltonians.
Time-dependent Hamiltonians arise in driven systems,
quantum control, and adiabatic
computation~\cite{Eck17,BCR10,FGGS00,
AVK+07},
and may also be introduced algorithmically through transformations such as the
interaction picture~\cite{LowWiebe18}. They further serve as primitives in
quantum algorithms for nonautonomous differential equations and linear
nonunitary dynamics~\cite{BC24,ALL23,CJL25}.
This leads to the general task of simulating a \emph{time-dependent} Hamiltonian
$H(t)$ on $t\in[0,T]$. The target operation is the time-ordered
propagator
\begin{equation*}
    U_H(T)
    =
    \mathcal{T}\exp\!\left(
        -i\int_0^T H(t)\,\mathrm{d}t
    \right),
\end{equation*}
where $\mathcal{T}$ orders the Hamiltonians chronologically.

Determining the query complexity of time-dependent Hamiltonian simulation is a fundamental problem in quantum algorithm design that has long remained open, while the optimal query complexity of time-independent Hamiltonian simulation is fully characterized.
For $\|H\|\leq\alpha$ and evolution scale $\alpha T\geq1$, the optimal
query complexity for achieving error at most $\eps$ is
\begin{align}\label{eq:timeIndependentBound}
\Theta\!\left(
        \alpha T
        +
        \frac{\log(1/\eps)}
        {\log\!\bigl(e+\log(1/\eps)/(\alpha T)\bigr)}
    \right)
\end{align}
queries~\cite{LC19,GSLW19}. This leaves open whether general time dependence
requires additional queries or whether the time-independent bound remains
achievable.

Existing methods do not attain the bound \eqref{eq:timeIndependentBound} for a general time-dependent
Hamiltonian.
Truncated Dyson-series algorithms partition the time interval into bounded-norm
segments and approximate the segment propagators
separately~\cite{LowWiebe18,KieferovaEtAl19,BerryEtAl20}.
These methods exploit the factorial
convergence of the Dyson series and can scale with the integrated
normalization of the Hamiltonian. In the standard segmented construction,
however, the precision-dependent truncation query cost is incurred on every segment,
so it combines \emph{multiplicatively} with the evolution scale $\alpha T$ rather than
\emph{additively}.
Another family of methods represents time dependence by a time-independent
Hamiltonian on a larger clock or Floquet-Hilbert
space~\cite{WatkinsEtAl24,MizutaFujii23,Mizuta23,CJL25,LiWang25}.
Floquet methods give optimal or nearly-optimal bounds for certain periodic and
multiperiodic Hamiltonians, while discrete-clock methods apply more generally
but their known query bounds still contain additional precision-dependent
overhead. 
Consequently, neither approach yields a general algorithm matching 
the time-independent lower bound.
Recent work proves optimal gate and query lower bounds for
time-dependent simulation under classical term access~\cite{ZAH26}.
Its query lower bound concerns noncoherent access and therefore does
not resolve the coherent-oracle query complexity studied here.
In addition, we are not aware of any lower
bound for time-dependent simulation in the usual coherent setting stronger than the bound
in~\eqref{eq:timeIndependentBound}.

This motivates the central question:

\begin{quote}
     \emph{What is the optimal query complexity of simulating general time-dependent Hamiltonians?}
\end{quote}

In this paper, we determine the optimal query complexity and show that general time
dependence incurs \emph{no} asymptotic query overhead. Our algorithm uses a
transducer in the framework of Belovs, Jeffery, and Yolcu~\cite{BJY24} to
encode the complete sequence of time-ordered steps, rather than
approximating and composing the short-time propagators independently. Beyond resolving the
time-dependent problem, the same framework yields a query-optimal alternative
to qubitization for time-independent Hamiltonian simulation.

\subsection{Main results}

We use the time-indexed block-encoding-access model of Low and
Wiebe~\cite{LowWiebe18}, denoted by $\HAMT$ and formally defined in
Definition~\ref{def:HAMTOracle}.  Our main result is as follows.

\begin{theorem}[Informal version of Theorem~\ref{thm:main}]
\label{thm:informalMain}
Given $\HAMT$ access to a Lipschitz-continuous time-dependent Hamiltonian
$H(t)$ on $[0,T]$ satisfying $\norm{H(t)}\leq\alpha$, there is a quantum
algorithm that approximates $U_H(T)$ to error at most $\eps$ using
\[
    \bigO\!\left(
        \alpha T
        +
        \frac{\log(1/\eps)}
        {\log\!\bigl(e+\log(1/\eps)/(\alpha T)\bigr)}
    \right)
\]
\(\HAMT\) queries.
\end{theorem}

This matches the lower bound in \eqref{eq:timeIndependentBound}, which already
applies to time-independent instances. Consequently,
general time dependence incurs \emph{no asymptotic query overhead}: up to constant
factors, its query complexity is the same as that of optimal time-independent
Hamiltonian simulation.

Our upper bound applies to Lipschitz-continuous Hamiltonians without assuming
periodicity or locality.  The Lipschitz
constant determines how finely time must be sampled and therefore affects the
gate and ancilla complexities, but not the query bound.
Table~\ref{tab:comparison} compares representative results.

\begin{table}[htb]
\centering
\footnotesize
\begingroup
\renewcommand{\arraystretch}{1.22}
\begin{tabular}{@{}>{\raggedright\arraybackslash}m{0.14\linewidth}
                    >{\raggedright\arraybackslash}m{0.20\linewidth}
                    >{\raggedright\arraybackslash}m{0.34\linewidth}
                    >{\raggedright\arraybackslash}m{0.24\linewidth}@{}}
\hline
\textbf{Work} & \textbf{Approach} & \textbf{Query complexity} & \textbf{Additional structure} \\
\hline
\cite{LowWiebe18,KieferovaEtAl19}
& Truncated Dyson series
&
  \(\displaystyle
  \bigO\biggl(
  \alpha T\,\frac{\log(\alpha T/\eps)}
  {\log\log(\alpha T/\eps)}
  \biggr)\) & -- \\
\cite{Mizuta23}
& Floquet-Hilbert-space embedding and qubitization
& \(\displaystyle
  \bigO\biggl(
  \alpha T+\frac{\log(1/\eps)}
  {\log\log\log(1/\eps)}
  \biggr)\) & Multiperiodic with \(O(1)\) periods and Fourier modes,
  and frequency scale \(O(1)\)\\
\cite{WatkinsEtAl24}
& Discrete-clock embedding and qubitization
& \(\displaystyle \widetilde{\mathcal O}\biggl(
  \alpha T+\log(1/\eps)+\frac{\beta T^{2}}{\eps}
  \biggr)\) & -- \\
\cite{LiWang25}
& Discrete clock with Gaussian quadrature
& \(\displaystyle \bigO\biggl(
  \alpha T\,\frac{\log(\alpha T/\eps)}
  {\log\log(\alpha T/\eps)}
  \biggr)\) & -- \\
This work
& Transducer-based LCU simulation
& \(\displaystyle \bigO\biggl(\alpha T+
  \frac{\log(1/\eps)}
  {\log\!\bigl(e+\log(1/\eps)/(\alpha T)\bigr)}
  \biggr)\) & --\\
\hline
\end{tabular}
\endgroup
\caption{Representative query complexities for time-dependent Hamiltonian
simulation.  Here \(\alpha\geq\max_{t\in[0,T]}\|H(t)\|\) is the oracle
normalization and \(\beta\) satisfies
\(\|H(t)-H(s)\|\leq\beta|t-s|\).  The Floquet row displays the precision
dependence for fixed \(\alpha T\), while the full joint bound is given by
Mizuta~\cite{Mizuta23}.  For the discrete-clock bound,
\(\widetilde{\mathcal O}\) suppresses multiplicative \(\log T\) and
\(\log\log(1/\eps)\) factors.}
\label{tab:comparison}
\end{table}

\medskip
\noindent
\textbf{Alternative time-independent construction.}
For $H(t)\equiv H$, the same framework also applies to time-independent
Hamiltonian simulation.  Rather than using qubitization, it takes
a weighted combination of
transducer-reuse circuits with different reuse lengths, causing the errors
from omitting the auxiliary state to cancel.  It therefore gives a different
query-optimal algorithm for time-independent Hamiltonian simulation.

\subsection{Related work}

\paragraph{Time-dependent product formulas.}
Product-formula methods for nonautonomous Schr\"odinger evolution were
introduced and analyzed by Huyghebaert and De
Raedt~\cite{HuyghebaertDeRaedt90}.  Higher-order decompositions of time-ordered
exponentials were later developed in~\cite{WiebeEtAl10}, and a time-dependent
Suzuki--Trotter expansion was used to treat arbitrarily rapidly varying local
Hamiltonians~\cite{PoulinEtAl11}.  For unbounded Hamiltonians,
state-dependent vector-norm bounds and commutator scaling for time-dependent
Trotter and generalized Trotter methods were established
in~\cite{AnFangLin21}.  More recent work derives explicit error bounds with
commutator scaling for time-dependent product and
multi-product formulas~\cite{MizutaIkedaFujii24}.  Time-dependent
multi-product formulas also arise from discrete-clock
constructions~\cite{WatkinsEtAl24}.  For low-energy inputs,
the simulation costs were reduced for time-dependent spin Hamiltonians
under a low-energy-support assumption~\cite{ZhouEtAl26}.  These guarantees
often exploit decomposability, commutator or locality structure, separated
energy scales, or a low-energy promise.  Our result instead addresses
worst-case black-box query complexity in the $\HAMT$ model without such
structure.

\paragraph{Dyson-series methods.}
A time-dependent extension based on a truncated Dyson series was suggested
alongside the truncated-Taylor-series method~\cite{BerryEtAl15} and later
developed into explicit algorithms that coherently combine Hamiltonian
samples at ordered times~\cite{LowWiebe18,KieferovaEtAl19}.  A rescaling
construction replaces the dependence on the maximum normalization by its time
integral~\cite{BerryEtAl20}.  By contrast, our algorithm combines
transducer-reuse circuits rather than terms in a truncated Dyson series.

\paragraph{Magnus-expansion methods.}
The qHOP algorithm uses a first-order Magnus approximation to simulate highly
oscillatory dynamics with commutator scaling and, for the Schr\"odinger
equation, exhibits second-order superconvergence~\cite{AnFangLin22}.  A
second-order Magnus algorithm has only logarithmic dependence on time
derivatives and exhibits fourth-order superconvergence for the Schr\"odinger
equation in the interaction picture~\cite{FangLiuSarkar25}.  A later fully
discrete analysis proves a superconvergence estimate uniform in the spatial
resolution~\cite{BornsWeilFangZhang26}.  High-order truncated Magnus algorithms
extend commutator scaling to arbitrary order while retaining only logarithmic
dependence on the Hamiltonian's time variation~\cite{FangLiuZhu25}.  These
methods can exploit small commutators while remaining effective for rapidly
varying Hamiltonians.

\paragraph{Clock and Floquet methods.}
Clock embeddings represent the dynamics by a time-independent Hamiltonian on
an enlarged Hilbert space.  Floquet-Hilbert-space methods give optimal or
nearly-optimal bounds for certain periodic and multiperiodic
Hamiltonians~\cite{MizutaFujii23,Mizuta23}.  Related constructions use an
enlarged discrete clock~\cite{WatkinsEtAl24}.  The continuous Sambe--Howland
clock provides a general framework for several digital and analog clock
constructions~\cite{CJL25}.  Li and Wang separate the clock and system
evolutions and approximate the resulting integrals by Gaussian
quadrature~\cite{LiWang25}.  This approach reduces the
precision dependence of discrete-clock simulation from polynomial to
polylogarithmic, although this factor still multiplies the evolution scale.
Across these approaches, a finite-dimensional implementation requires
discretizing or truncating the added clock or Fourier coordinate, which can
introduce additional precision-dependent overhead.  By contrast, the time
register in our construction only labels catalyst
components acted on in parallel by one \(\HAMT\) query and is not evolved by a
clock Hamiltonian.

\paragraph{Transducers.}
Transducers were introduced by Belovs, Jeffery, and Yolcu as an abstraction
of quantum algorithms for general state conversion and efficient subroutine
composition~\cite{BJY24}, building on the quantum Las Vegas query complexity
framework of Belovs and Yolcu~\cite{BY23}.  An expository account is given
in~\cite{Jef24}.  Recent applications use transducer-based purification to
obtain perfect completeness with an infinite counter~\cite{JW26} and
transducers to improve quantum-walk algorithms for sampling and
estimation~\cite{ARZ26}.  In
the present work, we construct a
transducer that uses one \(\HAMT\) query to implement a discretized
approximation to the time-ordered propagator \(U_H(T)\) when provided with an
auxiliary state, called the catalyst, and returns that state unchanged.  The general
reuse construction of Belovs, Jeffery, and Yolcu removes the need to prepare
this auxiliary state, but its generic norm estimate gives an
\(O(1/\eps^2)\) dependence on precision~\cite[Theorem~3.2]{BJY24}.  For the
transducer constructed here, we improve this to the optimal logarithmic precision dependence on precision in Theorem~\ref{thm:informalMain}.

\subsection{Discussion and future directions}

\paragraph{Improving gate complexity.}
Theorem~\ref{thm:informalMain} gives the optimal query complexity, but the
circuit implementation nevertheless has an additional gate cost: each application of the
one-query transducer performs one fixed update for every point of the finite
time grid, so the gate count grows linearly with the number of time steps.  A
higher-order time discretization might reduce the number of steps needed to
reach a given accuracy.

Appendix~\ref{app:continuousClock} gives another possible direction.  It
constructs a transducer on a continuous time register, which can be viewed as
the limit in which the finite time-step size tends to zero.  This removes the
error caused by approximating the evolution with finitely many short-time
steps, although the Hamiltonian must still be sampled on a finite grid.  The
clock states needed by the algorithm lie in a finite-dimensional
subspace, so replacing the continuous register by a finite register introduces
no further approximation error.  This construction retains the optimal query
complexity.  However, no efficient gate implementation of the resulting
finite-dimensional transducer is known, so this approach does not currently
improve the gate complexity of the direct construction.

\paragraph{Adiabatic and low-energy simulation.}
The lower bound matched by our algorithm concerns worst-case simulation of
the full propagator.  Adiabatic state preparation poses a different task:
the goal is to follow a selected low-energy state or subspace rather than
approximate the full propagator on arbitrary inputs.  Recent
low-energy analyses of time-dependent product formulas demonstrate that such
a promise can reduce the simulation cost~\cite{ZhouEtAl26}.  It is therefore
natural to ask whether the present method can similarly benefit from this
low-energy promise.  This would require showing that, for inputs in the
relevant subspace, both the catalyst norm and the error caused by omitting the
catalyst are controlled by the evolution within that subspace rather than by
the norm of the full Hamiltonian.

\paragraph{Beyond Hamiltonian simulation.}
Transducers provide a general framework for composing quantum algorithms, so
improving the general procedure for removing their catalysts could have
applications beyond Hamiltonian simulation.  The existing estimate for this
procedure applies to every transducer and therefore does not use the additional
structure present in our construction.  Here, combining different reuse
lengths cancels part of the error caused by omitting the catalyst.  Every term
that remains contains transitions that preserve the order of the time labels,
and transitions across widely separated labels are exponentially suppressed.
The resulting ordered sums give the factorial decay.  It
would be useful to identify general conditions under which a transducer's
catalyst can be removed with only logarithmic overhead in \(1/\eps\).

\section{Technical overview}
\label{sec:overview}

The construction has three stages.  First, we construct a unitary that,
when given an auxiliary vector \(\Gamma\ket\psi\), uses one \(\HAMT\)
query to implement \(U_C\), an approximation to \(U_H(T)\), and returns
the auxiliary vector unchanged.  In the terminology of Belovs,
Jeffery, and Yolcu~\cite{BJY24}, this unitary is a transducer and the
auxiliary vector is its catalyst.  The catalyst is generally difficult
to prepare, so this construction alone does not give an algorithm for
\(U_C\).  We therefore use their reuse construction, which
arranges \(N\) calls to the transducer so that an implementation supplied with the
scaled catalyst \(\Gamma\ket\psi/\sqrt N\) produces the desired output.
Omitting this scaled vector changes the output by \(O(N^{-1/2})\),
which gives a quadratic dependence on \(1/\eps\).  Our main additional
step is to derive a more precise expression for this error and choose a
weighted combination of different reuse lengths so that their errors
partly cancel.  The contributions that survive this cancellation must
follow the order of the short-time steps, and their total size decreases
factorially with the largest reuse length.  This gives the optimal
precision dependence in Theorem~\ref{thm:informalMain}.

\paragraph{Implementing the evolution with one query and a catalyst.}
The first stage uses the transducer formulation of Belovs, Jeffery,
and Yolcu~\cite[Sec.~3.1]{BJY24}.

\begin{definition}[Transducer]
Let \(\mathcal H_{\mathrm{pub}}\) and
\(\mathcal H_{\mathrm{priv}}\) be Hilbert spaces, not necessarily of
the same dimension, and let \(U\) be a unitary on
\(\mathcal H_{\mathrm{pub}}\).  A unitary \(S\) on
\(\mathcal H_{\mathrm{pub}}\oplus\mathcal H_{\mathrm{priv}}\) is a
transducer implementing \(U\) if there exists a linear map
\(\Gamma\colon\mathcal H_{\mathrm{pub}}\to
\mathcal H_{\mathrm{priv}}\) such that, for every vector
\(\ket\psi\in\mathcal H_{\mathrm{pub}}\),
\begin{equation}
 S\bigl(\ket\psi\oplus\Gamma\ket\psi\bigr)
 =U\ket\psi\oplus\Gamma\ket\psi.
\end{equation}
This is an identity of vectors, so
\(\ket\psi\oplus\Gamma\ket\psi\) need not be normalized.
\end{definition}
The vector \(\Gamma\ket\psi\) is called the catalyst because it enables
\(S\) to implement \(U\) and is returned unchanged.

To construct \(S\), we first approximate \(U_H(T)\) by a product of
\(J\) short-time unitaries.  We divide \([0,T]\) into \(J\)
intervals with endpoints \(t_j=jT/J\), for \(j=0,\ldots,J\).  The
oracle for this grid is
\(\HAMT=\sum_{j=0}^{J-1}\ketbra{j}{j}_{\mathsf T}\otimes O_j\),
where \(\mathsf T\) is a time-index register and each \(O_j\) is a
Hermitian unitary that block-encodes \(H(t_j)/\alpha\).  In particular,
\(O_j^2=I\) and \(\HAMT^2=I\).  We approximate the evolution on interval \(j\) by
the Cayley transform of \(H(t_j)\):
\[
\begin{aligned}
 U_j^C:=\left(I-\frac{iT}{2J}H(t_j)\right)
   \left(I+\frac{iT}{2J}H(t_j)\right)^{-1}=U_H(t_{j+1},t_j)+\bigO\!\left(\frac{T^2}{J^2}\right).
\end{aligned}
\]
Starting from the initial state \(\ket{\psi_0}\), define
\[
 \ket{\psi_{j+1}}=U_j^C\ket{\psi_j},\qquad
 U_C=U_{J-1}^C\cdots U_0^C.
\]
Thus \(\ket{\psi_J}=U_C\ket{\psi_0}\).  Accumulating the \(J\)
local errors gives
\[
 \norm{U_C-U_H(T)}
 \leq \frac{\beta T^2}{2J}+\frac{(\alpha T)^3}{12J^2}.
\]
Increasing \(J\) therefore controls this approximation error.  Its value
affects the gate complexity but not the query complexity: the transducer
below uses a single \(\HAMT\) query.

We next construct a transducer for each \(U_j^C\).  For every \(j\),
there is a local catalyst \(\ket{x_j}\) such that, after the oracle
block \(O_j\) at time \(t_j\) has been applied, a unitary
\(R_j\) performs
\begin{equation}
 \ket{\psi_j}\oplus O_j\ket{x_j}
 \xrightarrow{\,R_j\,}
 \ket{\psi_{j+1}}\oplus\ket{x_j}.
 \label{eq:localUpdateOverview}
\end{equation}
The unitary \(R_j\) uses no further query; its construction is given in
Section~\ref{subsec:onequery}.

To implement the whole product \(U_C\), we associate the local
catalyst \(\ket{x_j}\) with time-index label \(j\).  The global
catalyst is
\begin{equation}
 \Gamma\ket{\psi_0}
 :=\sum_{j=0}^{J-1}\ket j_{\mathsf T}\ket{x_j}.
 \label{eq:catalystRegister}
\end{equation}
The construction satisfies \(\norm{\Gamma}^2\leq\alpha T\); this
bound controls the cost of removing the catalyst below.
By the definition above, a single \(\HAMT\) query applies \(O_j\) to
the component with time-index label \(j\).  Its action on the global
catalyst is
\[
 \HAMT\bigl(\Gamma\ket{\psi_0}\bigr)
 =\sum_{j=0}^{J-1}\ket j_{\mathsf T}
   O_j\ket{x_j}.
\]
We then apply \(R_j\), for \(j=0,\ldots,J-1\), in increasing order,
with each \(R_j\) acting on the public state and the private component
with time-index label \(j\).  By \eqref{eq:localUpdateOverview}, these
operations successively transform the public state from
\(\ket{\psi_0}\) to \(\ket{\psi_J}=U_C\ket{\psi_0}\) while
restoring every \(\ket{x_j}\), and hence restore the global catalyst
\(\Gamma\ket{\psi_0}\).  Thus the \(\HAMT\) query followed by
\(R_0,\ldots,R_{J-1}\) defines a unitary \(S\) that uses one
\(\HAMT\) query and satisfies
\begin{equation}
 S\bigl(\ket{\psi_0}\oplus\Gamma\ket{\psi_0}\bigr)
 =U_C\ket{\psi_0}\oplus\Gamma\ket{\psi_0}.
 \label{eq:overviewTransducer}
\end{equation}

\paragraph{Finite reuse without a catalyst input.}
The second stage uses the reuse construction of Belovs, Jeffery,
and Yolcu~\cite[Theorem~3.2]{BJY24}.  It arranges \(N\) calls to \(S\) into a circuit
that implements \(U_C\) when supplied with only the scaled catalyst
\(\Gamma\ket{\psi_0}/\sqrt N\).

Figure~\ref{fig:catalystClock} compares the circuit supplied with the
scaled catalyst with the circuit whose private input is zero.  In
panel~(a), the scaled catalyst
\(\Gamma\ket{\psi_0}/\sqrt N\) is supplied as the private input to the first
call.  Let \(\mathsf K\) be the register that labels the order of the
\(N\) calls.  The operation \(F_N\) distributes the public input
uniformly over its labels:
\[
 \ket0_{\mathsf K}\ket{\psi_0}
 \xrightarrow{\,F_N\,}
 \frac1{\sqrt N}\sum_{\ell=0}^{N-1}
 \ket\ell_{\mathsf K}\ket{\psi_0}.
\]
The \(\ell\)th call acts on the public component with
\(\mathsf K=\ell\).  For \(\ell=0\), linearity of the transducer identity
\eqref{eq:overviewTransducer} gives
\[
 S\left(
 \frac{\ket{\psi_0}}{\sqrt N}
 \oplus\frac{\Gamma\ket{\psi_0}}{\sqrt N}
 \right)
 =\frac{U_C\ket{\psi_0}}{\sqrt N}
 \oplus\frac{\Gamma\ket{\psi_0}}{\sqrt N}.
\]
Thus the first call transforms the component with \(\mathsf K=0\) into
\(U_C\ket{\psi_0}/\sqrt N\) and returns the scaled catalyst unchanged, so
it can be supplied to the next call.  Repeating the same
identity for \(\ell=1,\ldots,N-1\) transforms the component with
\(\mathsf K=\ell\) into \(U_C\ket{\psi_0}/\sqrt N\) at each call, yielding
\[
 \frac1{\sqrt N}\sum_{\ell=0}^{N-1}
 \ket\ell_{\mathsf K}U_C\ket{\psi_0}.
\]
Uncomputing the register \(\mathsf K\) with \(F_N^\dagger\) gives the
final public output \(\ket0_{\mathsf K}U_C\ket{\psi_0}\).
Panel~(b) shows the same reuse circuit with its private input set to
zero.  Let \(\ket0_{\mathsf K}P_N\ket{\psi_0}\) denote its final public
component.  In standard block-encoding terminology~\cite{GSLW19}, the reuse
unitary block-encodes \(P_N\): projecting its auxiliary input and output
onto the all-zero state gives \(P_N\).

\begin{figure}[H]
\centering
\scalebox{0.86}{%
\begin{tikzpicture}[
  x=0.74cm,
  y=0.78cm,
  font=\small,
  >=Latex,
  panel/.style={
    draw=black!14,
    rounded corners=3pt,
    fill=black!1.2,
    line width=0.5pt
  },
  paneltitle/.style={
    font=\small\bfseries,
    text=black!82,
    anchor=west
  },
  transducer/.style={
    draw=black!58,
    rounded corners=2.4pt,
    fill=white,
    minimum width=1.70cm,
    minimum height=1.28cm,
    line width=0.68pt,
    font=\small\bfseries,
    align=center
  },
  prep/.style={
    draw=black!50,
    rounded corners=8pt,
    fill=white,
    minimum width=1.12cm,
    minimum height=0.60cm,
    line width=0.62pt,
    font=\small
  },
  public/.style={
    draw=black!64,
    line width=0.72pt,
    -{Latex[length=1.8mm,width=1.15mm]}
  },
  publicbus/.style={
    draw=black!54,
    line width=0.62pt
  },
  private/.style={
    draw=blue!66!black,
    line width=1.00pt,
    -{Latex[length=1.9mm,width=1.2mm]}
  },
  target/.style={
    draw=black!64,
    line width=0.72pt,
    -{Latex[length=1.9mm,width=1.2mm]}
  },
  smalllabel/.style={
    font=\scriptsize,
    text=black!68,
    align=center,
    inner sep=1pt
  },
  publiclabel/.style={
    font=\scriptsize,
    text=black!68,
    fill=black!1.2,
    inner xsep=1.5pt,
    inner ysep=0.8pt
  },
  privlabel/.style={
    font=\scriptsize,
    text=blue!66!black,
    fill=black!1.2,
    inner xsep=1.5pt,
    inner ysep=0.8pt
  },
  targetlabel/.style={
    font=\scriptsize,
    text=black!68,
    align=center,
    inner sep=1pt
  },
  note/.style={
    font=\scriptsize,
    text=black!58,
    align=center,
    inner sep=1pt
  }
]

\begin{scope}[shift={(0,7.20)}]
  \node[panel, minimum width=16.25cm, minimum height=5.6cm,
        anchor=south west] at (-0.4,0) {};
  \node[paneltitle] at (0,6.65) {(a) With the scaled catalyst};

  \node[smalllabel, anchor=south west] at (0.13,5)
    {$\ket0_{\mathsf K}\ket{\psi_0}$};
  \node[prep] (Fa) at (2.55,4.85) {$F_N$};
  \draw[public] (0,4.85) -- (Fa.west);
  \draw[publicbus] (Fa.east) -- (15.35,4.85);
  \node[note, anchor=south] at (9,5)
    {$\displaystyle \sum_{\ell=0}^{N-1}
      \ket{\ell}_{\mathsf K}\ket{\psi_0}/\sqrt{N}$};

  \node[transducer] (a0) at (5.30,2.75)
    {$S$\\[-2pt]\normalfont\scriptsize $\mathsf K=0$};
  \node[transducer] (a1) at (9.40,2.75)
    {$S$\\[-2pt]\normalfont\scriptsize $\mathsf K=1$};
  \node[transducer] (aN) at (15.35,2.75)
    {$S$\\[-2pt]\normalfont\scriptsize $\mathsf K=N-1$};

  \foreach \x/\boxname in {5.30/a0,9.40/a1,15.35/aN}{
    \draw[public] (\x,4.85) -- (\boxname.north);
    \node[publiclabel, anchor=base east]
      at ({\x-0.10},3.94)
      {$\ket{\psi_0}/\sqrt N$};
  }
  \node[note, anchor=base] at (12.35,3.94) {$\cdots$};

  \draw[public] (a0.south) -- (5.30,0.90);
  \draw[public] (a1.south) -- (9.40,0.90);
  \draw[public] (aN.south) -- (15.35,0.90);
  \draw[publicbus] (5.30,0.90) -- (15.35,0.90);

  \node[smalllabel, anchor=base] at (5.30,0.4)
    {$U_C\ket{\psi_0}/\sqrt N$};
  \node[smalllabel, anchor=base] at (9.40,0.4)
    {$U_C\ket{\psi_0}/\sqrt N$};
  \node[note, anchor=base] at (12.35,0.4) {$\cdots$};
  \node[smalllabel, anchor=base] at (15.35,0.4)
    {$U_C\ket{\psi_0}/\sqrt N$};

  \node[prep] (Fad) at (17.80,0.90) {$F_N^\dagger$};
  \draw[public] (15.35,0.90) -- (Fad.west);
  \draw[target] (Fad.east) -- (21.0,0.90);
  \node[targetlabel, anchor=south] at (19.7,1.05)
    {$\ket0_{\mathsf K}\,U_C\ket{\psi_0}$};

  \draw[private] (0,2.75) -- (a0.west);
  \draw[private] (a0.east) -- (a1.west);
  \draw[private] (a1.east) -- (11.80,2.75);
  \draw[private] (12.90,2.75) -- (aN.west);
  \draw[private] (aN.east) -- (21,2.75);

  \node[privlabel, anchor=south] at (2.1,2.90)
    {$\Gamma\ket{\psi_0}/\sqrt N$};
  \node[note, text=blue!66!black, fill=black!1.2]
    at (12.35,2.75) {$\cdots$};
  \node[privlabel, anchor=south] at (18.8,2.90)
    {$\Gamma\ket{\psi_0}/\sqrt N$};
\end{scope}

\begin{scope}
  \node[panel, minimum width=16.25cm, minimum height=5.6cm,
        anchor=south west] at (-0.4,0) {};
  \node[paneltitle] at (0,6.65)
    {(b) With zero private input};

  \node[smalllabel, anchor=south west] at (0.13,5)
    {$\ket0_{\mathsf K}\ket{\psi_0}$};
  \node[prep] (Fb) at (2.55,4.85) {$F_N$};
  \draw[public] (0,4.85) -- (Fb.west);
  \draw[publicbus] (Fb.east) -- (15.35,4.85);
  \node[note, anchor=south] at (9,5)
    {$\displaystyle \sum_{\ell=0}^{N-1}
      \ket{\ell}_{\mathsf K}\ket{\psi_0}/\sqrt{N}$};

  \node[transducer] (b0) at (5.30,2.75)
    {$S$\\[-2pt]\normalfont\scriptsize $\mathsf K=0$};
  \node[transducer] (b1) at (9.40,2.75)
    {$S$\\[-2pt]\normalfont\scriptsize $\mathsf K=1$};
  \node[transducer] (bN) at (15.35,2.75)
    {$S$\\[-2pt]\normalfont\scriptsize $\mathsf K=N-1$};

  \foreach \x/\boxname in {5.30/b0,9.40/b1,15.35/bN}{
    \draw[public] (\x,4.85) -- (\boxname.north);
    \node[publiclabel, anchor=base east]
      at ({\x-0.10},3.94)
      {$\ket{\psi_0}/\sqrt N$};
  }
  \node[note, anchor=base] at (12.35,3.94) {$\cdots$};

  \draw[public] (b0.south) -- (5.30,0.90);
  \draw[public] (b1.south) -- (9.40,0.90);
  \draw[public] (bN.south) -- (15.35,0.90);
  \draw[publicbus] (5.30,0.90) -- (15.35,0.90);

  \node[smalllabel, anchor=base] at (5.30,0.4) {$\ket{p_0}$};
  \node[smalllabel, anchor=base] at (9.40,0.4) {$\ket{p_1}$};
  \node[note, anchor=base] at (12.35,0.4) {$\cdots$};
  \node[smalllabel, anchor=base] at (15.35,0.4) {$\ket{p_{N-1}}$};

  \node[prep] (Fbd) at (17.80,0.90) {$F_N^\dagger$};
  \draw[public] (15.35,0.90) -- (Fbd.west);
  \draw[target] (Fbd.east) -- (21.0,0.90);
  \node[targetlabel, anchor=south] at (19.7,1.05)
    {$\ket0_{\mathsf K}\,P_N\ket{\psi_0}$};

  \draw[private] (0,2.75) -- (b0.west);
  \draw[private] (b0.east) -- (b1.west);
  \draw[private] (b1.east) -- (11.80,2.75);
  \draw[private] (12.90,2.75) -- (bN.west);
  \draw[private] (bN.east) -- (21,2.75);

  \node[privlabel, anchor=south] at (2.1,2.90) {$\ket{z_0}=0$};
  \node[privlabel, anchor=south] at (7.35,2.90) {$\ket{z_1}$};
  \node[note, text=blue!66!black, fill=black!1.2]
    at (12.35,2.75) {$\cdots$};
  \node[privlabel, anchor=south east] at (14.18,2.90)
    {$\ket{z_{N-1}}$};
  \node[privlabel, anchor=south] at (18.8,2.90) {$\ket{z_N}$};
\end{scope}

\end{tikzpicture}
}
\caption{State flow for finite reuse.  In panel~(a), the same scaled
catalyst is reused and the public outputs recombine into
\(U_C\ket{\psi_0}\).  Panel~(b) starts at \(\ket{z_0}=0\);
\(\ket{z_\ell}\) and \(\ket{p_\ell}\) are the private input and public
output of call \(\ell\).  After \(F_N^\dagger\), the \(\mathsf K=0\)
public component is \(P_N\ket{\psi_0}\).  Gray and blue lines denote
public and private components, respectively.}
\label{fig:catalystClock}
\end{figure}
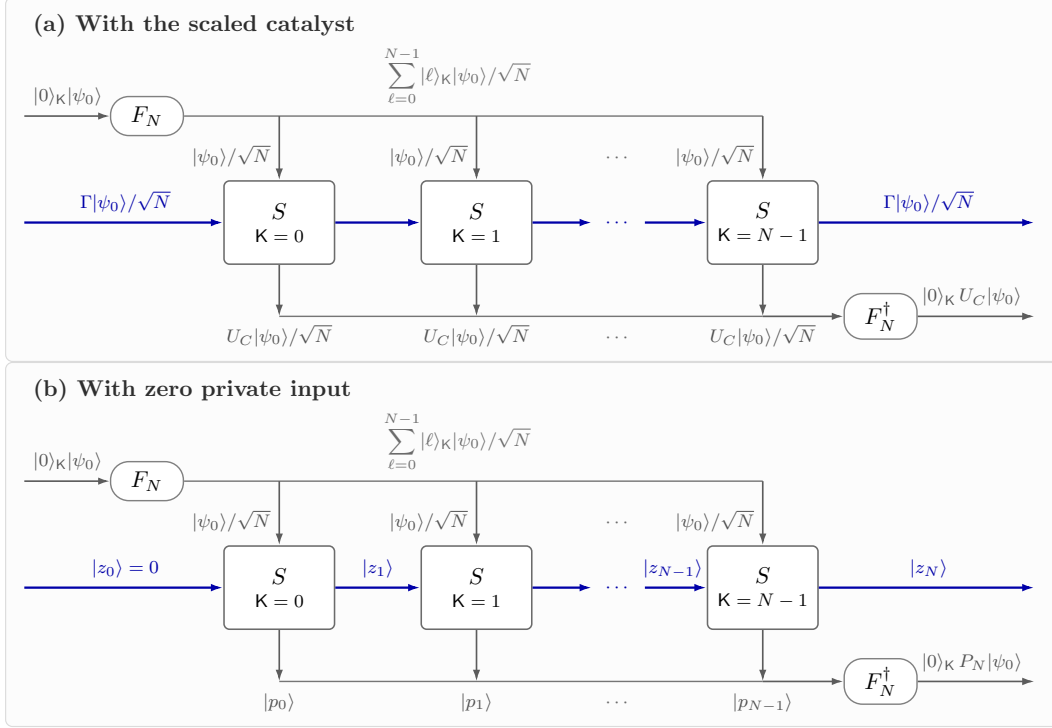

For a unit initial state \(\ket{\psi_0}\), the inputs in panels~(a)
and~(b) differ
only by the scaled catalyst, whose norm is at most
\(\norm{\Gamma}/\sqrt N\).  Unitarity preserves this distance, and taking
the public component cannot increase it.  Hence
\(\norm{U_C-P_N}\leq\norm{\Gamma}/\sqrt N\).  This estimate uses only
the transducer identity and unitarity, not any additional structure of
\(S\).  Thus the generic reuse analysis achieves
\(\norm{U_C-P_N}\leq\eps\) with
\[
 N=\bigO\!\left(1+\frac{\norm{\Gamma}^2}{\eps^2}\right)
 =\bigO\!\left(1+\frac{\alpha T}{\eps^2}\right).
\]
Its quadratic dependence on \(1/\eps\) is insufficient for the query
complexity in Theorem~\ref{thm:informalMain}.

\paragraph{Factorial error decay by combining reuse lengths.}
For the transducer constructed here, the effect of omitting the catalyst
can be described more precisely.  Define the private-to-private
block \(A\) and private-to-public block \(C\) by
\[
 S(0\oplus\ket v)=C\ket v\oplus A\ket v.
\]
Since \(S\) is unitary, both \(A\) and \(C\) are contractions.
For \(G_N(z):=N^{-1}\sum_{\ell=0}^{N-1}z^\ell\), the error is
\[
 U_C-P_N=C\,G_N(A)\Gamma
 =\frac1N\sum_{\ell=0}^{N-1}CA^\ell\Gamma.
\]
Here \(A^\ell\Gamma\ket{\psi_0}\) is the contribution obtained after
the omitted catalyst undergoes the private-to-private block \(A\) for
\(\ell\) successive calls, and \(C\) maps this contribution to the
public output on the next call.  The polynomial \(G_N\) averages these
contributions over the reuse calls.

We now combine reuse circuits with different lengths.  Choose real
coefficients satisfying
\(\sum_{N=1}^L\lambda_N=1\), and write
\(\widetilde U:=\sum_{N=1}^L\lambda_NP_N\),
\(\Lambda:=\sum_{N=1}^L|\lambda_N|\).  Applying the standard LCU
construction for block-encoded matrices~\cite{GSLW19} to the
reuse circuits produces a block-encoding of \(\widetilde U\) with
normalization \(\Lambda\).  Keeping \(\Lambda\leq2\) allows one step of
oblivious amplitude amplification to remove this normalization with
constant overhead~\cite{BerryEtAl15}.  Define
\(E(z):=\sum_{N=1}^L\lambda_NG_N(z)\). The identity
\(U_C-P_N=C\,G_N(A)\Gamma\) and the condition
\(\sum_N\lambda_N=1\) imply
\[
 U_C-\widetilde U=C\,E(A)\Gamma.
\]
Using \(\norm{C}\leq1\) and
\(\norm{\Gamma}\leq\sqrt{\alpha T}\), we obtain
\[
 \big\|U_C-\widetilde U\big\|
 \leq\sqrt{\alpha T}\,\norm{E(A)}.
\]
Thus the coefficients must make \(\norm{E(A)}\) small while keeping
\(\Lambda\) bounded.

To make \(\norm{E(A)}\) small, we choose \(E\) to contain a power of
\(Q(z):=(1+z^2)/2\).  We prove that
\[
 \norm{Q(A)^q}
 \leq\left(\frac{12e\alpha T}{q}\right)^q.
\]
The proof uses the representation \(A=i(I-\Delta)\HAMT\).  On the time
register \(\mathsf T\), \(\Delta\) maps a component with label \(j\)
only to labels \(r\geq j\), and the norms of its blocks decay
exponentially as \(r-j\) increases.  Since \(\HAMT^2=I\), this
representation implies
\[
 2Q(A)=I+A^2
 =\Delta+\HAMT\Delta\HAMT
   -\Delta\HAMT\Delta\HAMT.
\]
In the expansion of \(I+A^2\), the \(I\) term cancels the part of
\(A^2\) containing no \(\Delta\).  Hence every term in
\(Q(A)^q\) contains at least \(q\) factors of \(\Delta\), whose time
labels form a nondecreasing sequence.  Summing over these ordered
labels, together with the exponential block decay, proves the displayed
estimate.

The factor \(Q^q\) provides the required decay.  We now choose a
polynomial \(E_q\) that contains this factor while keeping its LCU
normalization bounded.  The coefficients \(\lambda_N\) in
\(E_q(z)=\sum_N\lambda_NG_N(z)\) determine the normalization
\(\Lambda=\sum_N|\lambda_N|\).  Writing the same polynomial as
\(E_q(z)=\sum_ke_kz^k\), the change-of-basis relation
\(\lambda_N=N(e_{N-1}-e_N)\) shows that \(\Lambda\) is the weighted
variation between neighboring monomial coefficients.  Multiplication
by \(G_{2q}\) averages consecutive monomial coefficients and reduces
this variation.  We therefore take
\[
 E_q(z):=Q(z)^qG_{2q}(z)
 =\sum_{N=1}^{4q}\lambda_NG_N(z),\qquad
 \sum_N\lambda_N=1,\quad
 \Lambda=\sum_N|\lambda_N|<2.
\]
Since \(A\) is a contraction,
\(\norm{G_{2q}(A)}\leq1\).  Therefore, substituting
\(E_q(A)=Q(A)^qG_{2q}(A)\) gives
\[
 \left\|U_C-\sum_{N=1}^{4q}\lambda_NP_N\right\|
 \leq\sqrt{\alpha T}\,
 \left\|Q(A)^q\right\|\left\|G_{2q}(A)\right\|
 \leq\sqrt{\alpha T}
 \left(\frac{12e\alpha T}{q}\right)^q.
\]
The LCU construction uses only reuse lengths \(N\leq4q\) and, after
amplification, \(12q\) queries in total.

\section{Preliminaries and oracle access}
\label{sec:preliminaries}

We now give the register and oracle definitions used in the overview
and in the formal construction below.  Let \(\mathsf S\) be an
\(n\)-qubit system register with Hilbert space \(\Sys\), and let
\(H\colon[0,T]\to\operatorname{Herm}(\Sys)\) satisfy
\(\norm{H(t)}\leq\alpha\) and
\(\norm{H(t)-H(s)}\leq\beta|t-s|\) for all \(s,t\in[0,T]\), where
\(\alpha>0\) and \(\beta\geq0\).  Let
\(U_H(t)\) denote the
propagator satisfying
\[
 i\frac{\mathrm{d}}{\mathrm{d}t}U_H(t)=H(t)U_H(t),\qquad U_H(0)=I.
\]
For \(0\leq s\leq t\leq T\), write
\(U_H(t,s):=U_H(t)U_H(s)^\dagger\).

\begin{samepage}
We use the standard block-encoding terminology of
Gily\'en, Su, Low, and Wiebe~\cite{GSLW19}.
\begin{definition}[Block-encoding]
Let \(\mathsf R\) be an \(a\)-qubit auxiliary register, and let
\(\ket0_{\mathsf R}\) denote its all-zero state.  For \(\alpha>0\) and
\(\epsilon\geq0\), a unitary \(V\) on
\(\mathcal H_{\mathsf R}\otimes\Sys\) is an
\((\alpha,a,\epsilon)\)-block-encoding of a linear operator \(X\) on \(\Sys\)
if
\[
 \left\|X-\alpha
 (\bra0_{\mathsf R}\otimes I)V
 (\ket0_{\mathsf R}\otimes I)\right\|\leq\epsilon.
\]
When \(\alpha=1\) and \(\epsilon=0\), we simply say that \(V\)
block-encodes \(X\).
\end{definition}
\end{samepage}

For a power of two \(J\) chosen by the algorithm, define
\[
 t_j:=\frac{jT}{J},\qquad j=0,\ldots,J.
\]
Let \(\mathsf A\) be an \(a\)-qubit ancilla register with Hilbert
space \(\Anc=(\mathbb C^2)^{\otimes a}\) and initial state
\(\ket0_{\mathsf A}\).  Let \(\mathsf T\) be the time-index register
with Hilbert space \(\Time=\mathbb C^J\); its basis state
\(\ket j_{\mathsf T}\) labels the physical time \(t_j\).

We use the \(\HAMT\) access model introduced by Low and
Wiebe~\cite{LowWiebe18}.  Our Cayley construction additionally requires
each time-indexed block encoding to be a Hermitian unitary.\footnote{This
additional requirement can be enforced with one flag qubit.  If \(U\)
block-encodes the Hermitian operator \(H(t_j)/\alpha\), then
\(\widetilde U=\ketbra{0}{1}\otimes U+
\ketbra{1}{0}\otimes U^\dagger\) is Hermitian and unitary.  Its
compression onto \(\ket+\) in the flag register is
\((U+U^\dagger)/2\), which still block-encodes \(H(t_j)/\alpha\).}
\begin{definition}[\(\HAMT\) oracle]
\label{def:HAMTOracle}
 For the grid above,
the oracle is the Hermitian unitary \(\HAMT_J\) on
\(\Time\otimes\Anc\otimes\Sys\) defined by
\begin{equation}
 \HAMT_J
 =\sum_{j=0}^{J-1}\ketbra{j}{j}_{\mathsf T}
  \otimes O_j,
 \label{eq:HAMT}
\end{equation}
where each \(O_j\) is a Hermitian unitary on
\(\Anc\otimes\Sys\) satisfying
\begin{equation}
 (\bra0_{\mathsf A}\otimes I)O_j
 (\ket0_{\mathsf A}\otimes I)=\frac{H(t_j)}{\alpha}.
 \label{eq:OjBlock}
\end{equation}
Thus \(O_j\) is a Hermitian block-encoding of \(H(t_j)\), with
parameters \((\alpha,a,0)\).
\end{definition}

We use this \(J\) throughout the circuit and suppress the subscript on
\(\HAMT_J\).  A controlled application of \(\HAMT\) counts as one
query.

\section{The Cayley transducer}

We first approximate the original evolution by a product of \(J\) Cayley
transforms.  We then implement the entire product with one transducer
that uses a single oracle query and returns its catalyst unchanged.

\subsection{Cayley step}
\label{subsec:cayley}

For the power-of-two step count \(J\) fixed above, this subsection
defines the Cayley steps and bounds the error of their product relative
to \(U_H(T)\).
Set
\[
 \delta=\frac{T}{J},\qquad w=\alpha\delta,
\]
and assume \(0<w\leq1\).  Given \(\ket{\psi_j}_{\mathsf S}\), define
the midpoint \(\ket{y_j}\) and the next state $\ket{\psi_{j+1}}$ by
\begin{equation}
 \ket{y_j}:=\frac{\ket{\psi_j}+\ket{\psi_{j+1}}}2,
 \qquad
 \ket{\psi_{j+1}}-\ket{\psi_j}
 =-\frac{iw}{\alpha}H(t_j)\ket{y_j}.
 \label{eq:midpointRule}
\end{equation}
Rearranging \eqref{eq:midpointRule} gives
\begin{equation}
 \begin{aligned}
 \ket{\psi_j}
 &=\left(I+\frac{iw}{2\alpha}H(t_j)\right)\ket{y_j},
 &\qquad
 \ket{\psi_{j+1}}
 &=\left(I-\frac{iw}{2\alpha}H(t_j)\right)\ket{y_j},\\
 \ket{y_j}
 &=\left(I+\frac{iw}{2\alpha}H(t_j)\right)^{-1}\ket{\psi_j},
 \end{aligned}
 \label{eq:twosides}
\end{equation}
so $\ket{\psi_{j+1}}$ and $\ket{y_j}$ are well-defined vectors because
$I + \frac{iw}{2\alpha}H(t_j)$ is nonsingular for Hermitian $H(t_j)$.

Using both relations in \eqref{eq:midpointRule}, this update preserves
the norm exactly:
\[
 \begin{aligned}
 \norm{\ket{\psi_{j+1}}}^2-\norm{\ket{\psi_j}}^2
 &=2\operatorname{Re}\,
   \bra{y_j}\bigl(\ket{\psi_{j+1}}-\ket{\psi_j}\bigr)\\
 &=2\operatorname{Re}\!\left(
   -\frac{iw}{\alpha}\bra{y_j}H(t_j)\ket{y_j}
   \right)=0.
 \end{aligned}
\]

Equivalently, the single-step update is the Cayley transform
\begin{equation}
 U_j^C=
 \left(I-\frac{iw}{2\alpha}H(t_j)\right)
 \left(I+\frac{iw}{2\alpha}H(t_j)\right)^{-1}.
 \label{eq:localCayley}
\end{equation}
On an eigenspace of \(H(t_j)/\alpha\) with eigenvalue \(\lambda\),
\[
 \frac{1-iw\lambda/2}{1+iw\lambda/2}
 =\exp\!\left(-2i\arctan\!\frac{w\lambda}{2}\right).
\]
The spectral theorem and
\(\lvert x-2\arctan(x/2)\rvert\leq\lvert x\rvert^3/12\) for
\(\lvert x\rvert\leq1\) give the explicit local error bound
\begin{equation}
 \norm{U_j^C-e^{-iwH(t_j)/\alpha}}\leq\frac{w^3}{12}.
 \label{eq:CayleyOrder}
\end{equation}

For an initial state \(\ket{\psi_0}_{\mathsf S}\), set
\begin{equation}
 \ket{\psi_{j+1}}=U_j^C\ket{\psi_j},
 \label{eq:qtrajectory}
\end{equation}
and define \(U_C\ket{\psi_0}=\ket{\psi_J}\).  On the \(j\)th interval,
the exponential generated by the constant Hamiltonian \(H(t_j)\)
approximates the exact short-time propagator.  Since \(H\) is
\(\beta\)-Lipschitz, Duhamel's formula gives
\[
 \norm{e^{-i\delta H(t_j)}-U_H(t_{j+1},t_j)}
 \leq\int_{t_j}^{t_{j+1}}\norm{H(t)-H(t_j)}\,\mathrm{d}t
 \leq\frac{\beta\delta^2}{2}.
\]
Combining this estimate with \eqref{eq:CayleyOrder} and telescoping
over the \(J\) intervals gives
\begin{equation}
 \norm{U_C-U_H(T)}
 \leq\frac{\beta T^2}{2J}
 +\frac{(\alpha T)^3}{12J^2}.
 \label{eq:CayleyEvolutionError}
\end{equation}

\subsection{One-query transducer}
\label{subsec:onequery}

We now construct a unitary that advances all \(J\) Cayley steps with
one query to the oracle when supplied with the catalyst defined below.

Set \(\Priv:=\Time\otimes\Anc\otimes\Sys\).  On the private space,
one \(\HAMT\) query acts as
\begin{equation}
 \HAMT\bigl(\ket j_{\mathsf T}\ket v_{\mathsf{AS}}\bigr)
 =\ket j_{\mathsf T}O_j\ket v_{\mathsf{AS}}.
 \label{eq:fineOracle}
\end{equation}
By \eqref{eq:OjBlock}, \(O_j\) is an
\((\alpha,a,0)\)-block-encoding of \(H(t_j)\).

To represent the public--private direct sum on registers, we use a flag
qubit.
\begin{convention}[Flag encoding]
\label{conv:publicPrivate}
In this construction, \(\mathsf P\) represents \(\Sys\oplus\Priv\) by
\[
 \ket\psi_{\mathsf S}\oplus\ket v_{\mathsf{TAS}}
 \quad\longleftrightarrow\quad
 \ket0_{\mathsf{PTA}}\ket\psi_{\mathsf S}
 +\ket1_{\mathsf P}\ket v_{\mathsf{TAS}},
\]
where
\(\ket0_{\mathsf{PTA}}
:=\ket0_{\mathsf P}\ket0_{\mathsf T}\ket0_{\mathsf A}\)\footnote{The
public state is padded with zero states on \(\mathsf T\) and \(\mathsf A\)
so that the direct sum can be encoded using the flag qubit \(\mathsf P\).}.
\end{convention}

Define the map that appends a zero ancilla and its adjoint by
\[
 \iota_0:=\ket0_{\mathsf A}\otimes I\colon
 \Sys\longrightarrow\Anc\otimes\Sys,
 \qquad
 \iota_0^\dagger=\bra0_{\mathsf A}\otimes I.
\]
The associated projector is
\[
 \Pi_{\mathsf A,0}
 :=\iota_0\iota_0^\dagger
 =\ketbra{0}{0}_{\mathsf A}\otimes I_{\mathsf S}.
\]
This projector selects the \(\ket0_{\mathsf A}\) component inside
\(\Anc\otimes\Sys\); it is unrelated to the value of \(\mathsf P\).
Define the local catalyst
\begin{equation}
 \ket{x_j}_{\mathsf{AS}}
 :=\sqrt{\frac w2}(I+iO_j)
 \iota_0\ket{y_j}.
 \label{eq:localcatalyst}
\end{equation}
Substituting \eqref{eq:localcatalyst} and using \(O_j^2=I\) gives
\begin{equation}
 \begin{aligned}
 \ket{x_j}-iO_j\ket{x_j}
 &=\sqrt{\frac w2}\,
 \bigl[(I+iO_j)-iO_j(I+iO_j)\bigr]
 \iota_0\ket{y_j}\\
 &=\sqrt{\frac w2}\,
 \bigl(I+iO_j-iO_j+O_j^2\bigr)
 \iota_0\ket{y_j}\\
 &=\sqrt{2w}\,\iota_0\ket{y_j}.
 \end{aligned}
 \label{eq:catalystDifference}
\end{equation}
Since \(O_j\) is a Hermitian unitary by
Definition~\ref{def:HAMTOracle}, \((I+iO_j)/\sqrt2\) is unitary, and
the catalyst norm satisfies
\begin{equation}
 \norm{\ket{x_j}}^2
 =w\norm{\ket{y_j}}^2
 \leq w\norm{\ket{\psi_j}}^2,
 \label{eq:stepcost}
\end{equation}
where the inequality follows from
\(\norm{(I+\frac{iw}{2\alpha}H(t_j))^{-1}}\leq1\).

Set \(c=(1-w/2)/(1+w/2)\) and
\(s=\sqrt{2w}/(1+w/2)\).
Because the step size is independent of \(j\), these coefficients do
not depend on \(j\).
Define the following operator on
\(\Sys\oplus(\Anc\otimes\Sys)\):
\begin{equation}
 R=
 \begin{pmatrix}
 cI
 &-is\iota_0^\dagger\\
 s\iota_0
 &i\bigl(c\Pi_{\mathsf A,0}+I-\Pi_{\mathsf A,0}\bigr)
 \end{pmatrix}.
 \label{eq:cayleyUpdate}
\end{equation}

\begin{lemma}[Local Cayley identity]
\label{lem:localCayley}
The operator \(R\) is unitary and, for every \(j\),
\begin{equation}
 R\bigl(\ket{\psi_j}\oplus O_j\ket{x_j}\bigr)
 =\ket{\psi_{j+1}}\oplus\ket{x_j}.
 \label{eq:stepgraph}
\end{equation}
\end{lemma}

\begin{proof}
We first verify that this operator is unitary.  Under the identification
of \(\Sys\) with \(\operatorname{ran}\Pi_{\mathsf A,0}\) through
\(\iota_0\), the restriction of \(R\) to
\(\Sys\oplus\operatorname{ran}\Pi_{\mathsf A,0}\) is represented by
\[
 \begin{pmatrix}c&-is\\ s&ic\end{pmatrix},
\]
whereas it acts as \(iI\) on the remaining private subspace
\(0\oplus(\operatorname{ran}\Pi_{\mathsf A,0})^\perp\).  Since
\(c^2+s^2=1\), \(R\) is unitary.

To evaluate the public and private outputs in \eqref{eq:stepgraph}, we
use
\[
 s=\sqrt{\frac w2}(1+c),\qquad
 1-c=\frac w2(1+c)=s\sqrt{\frac w2}.
\]
The block-encoding identity \eqref{eq:OjBlock}, together with
\eqref{eq:localcatalyst} and \(O_j^2=I\), yields
\begin{equation}
 \begin{aligned}
 \iota_0^\dagger O_j\ket{x_j}
 &=\sqrt{\frac w2}\,
 \iota_0^\dagger O_j(I+iO_j)
 \iota_0\ket{y_j}\\
 &=\sqrt{\frac w2}\,
 \bigl(\iota_0^\dagger O_j\iota_0+iI\bigr)
 \ket{y_j}\\
 &=\sqrt{\frac w2}
 \left(\frac{H(t_j)}{\alpha}+iI\right)\ket{y_j}.
 \end{aligned}
 \label{eq:queriedCatalystProjection}
\end{equation}
Using \eqref{eq:cayleyUpdate}, the first identity in
\eqref{eq:twosides}, and \eqref{eq:queriedCatalystProjection}, the
public output is
\[
 \begin{aligned}
 c\ket{\psi_j}-is\iota_0^\dagger O_j\ket{x_j}
 &=c\left(I+\frac{iw}{2\alpha}H(t_j)\right)\ket{y_j}
 -is\sqrt{\frac w2}
 \left(\frac{H(t_j)}{\alpha}+iI\right)\ket{y_j}\\
 &=\left[
 \left(c+s\sqrt{\frac w2}\right)I
 +i\left(\frac{cw}{2}-s\sqrt{\frac w2}\right)
 \frac{H(t_j)}{\alpha}\right]\ket{y_j}\\
 &=
 \left[I+i\left(\frac{cw}{2}-(1-c)\right)
 \frac{H(t_j)}{\alpha}\right]
 \ket{y_j}
 =\left(I-\frac{iw}{2\alpha}H(t_j)\right)\ket{y_j}
 =\ket{\psi_{j+1}}.
 \end{aligned}
\]
The scalar simplification uses
\(c+s\sqrt{w/2}=1\),
\(s\sqrt{w/2}=1-c\), and
\(cw/2-(1-c)=-w/2\), all of which follow from the definitions of
\(c\) and \(s\).  The final equality is the second identity in
\eqref{eq:twosides}.
For the private output, first project onto the subspace in which
\(\mathsf A\) is in \(\ket0\), and use the first identity in
\eqref{eq:twosides}
together with \eqref{eq:queriedCatalystProjection}:
\[
 \begin{aligned}
 \Pi_{\mathsf A,0}\left[
 s\iota_0\ket{\psi_j}
 +i\bigl(c\Pi_{\mathsf A,0}+I-\Pi_{\mathsf A,0}\bigr)
 O_j\ket{x_j}\right]
 &=s\iota_0\ket{\psi_j}
   +ic\Pi_{\mathsf A,0}O_j\ket{x_j}\\
 &=\iota_0\left[
 s\left(I+\frac{iw}{2\alpha}H(t_j)\right)
 +ic\sqrt{\frac w2}
 \left(\frac{H(t_j)}{\alpha}+iI\right)
 \right]\ket{y_j}\\
 &=\iota_0
 \left[
 \left(s-c\sqrt{\frac w2}\right)I
 +i\left(\frac{sw}{2}+c\sqrt{\frac w2}\right)
 \frac{H(t_j)}{\alpha}
 \right]\ket{y_j}\\
 &=\sqrt{\frac w2}
   \iota_0\left(I+i\frac{H(t_j)}{\alpha}\right)\ket{y_j}\\
 &=\sqrt{\frac w2}\,
   \Pi_{\mathsf A,0}(I+iO_j)\iota_0\ket{y_j}
 =\Pi_{\mathsf A,0}\ket{x_j}.
 \end{aligned}
\]
The coefficient reduction uses
\(s-c\sqrt{w/2}=\sqrt{w/2}\) and
\(sw/2+c\sqrt{w/2}=\sqrt{w/2}\).  The penultimate equality follows
from \eqref{eq:OjBlock}, since
\[
 \Pi_{\mathsf A,0}O_j\iota_0
 =\iota_0\bigl(\iota_0^\dagger O_j\iota_0\bigr)
 =\iota_0\frac{H(t_j)}{\alpha}.
\]
The final equality uses the definition of \(\ket{x_j}\) in
\eqref{eq:localcatalyst}.
On the orthogonal ancilla subspace,
applying \(I-\Pi_{\mathsf A,0}\) to
\eqref{eq:catalystDifference} and using
\((I-\Pi_{\mathsf A,0})\iota_0=0\) gives
\(i(I-\Pi_{\mathsf A,0})O_j\ket{x_j}
=(I-\Pi_{\mathsf A,0})\ket{x_j}\).  Hence the private output satisfies
\[
 \begin{aligned}
 (I-\Pi_{\mathsf A,0})
 \left[
 s\iota_0\ket{\psi_j}
 +i\bigl(c\Pi_{\mathsf A,0}+I-\Pi_{\mathsf A,0}\bigr)
 O_j\ket{x_j}\right]
 &=i(I-\Pi_{\mathsf A,0})O_j\ket{x_j}\\
 &=(I-\Pi_{\mathsf A,0})\ket{x_j}.
 \end{aligned}
\]
Thus the private output is \(\ket{x_j}\), while
the public output is \(\ket{\psi_{j+1}}\), which proves
\eqref{eq:stepgraph}.
\end{proof}

For each \(j\), let \(R_j\) apply the unitary \(R\) from
\eqref{eq:cayleyUpdate} jointly to the public component and the private
component with time-index label \(j\), while
leaving the private components with labels \(r\ne j\),
unchanged.  Under
Convention~\ref{conv:publicPrivate}, its register action is
\begin{equation}
 \begin{aligned}
 \ket0_{\mathsf{PTA}}\ket\psi_{\mathsf S}
 &\longmapsto
 c\ket0_{\mathsf{PTA}}\ket\psi
 +s\ket1_{\mathsf P}\ket j_{\mathsf T}
   \ket0_{\mathsf A}\ket\psi,\\
 \ket1_{\mathsf P}\ket j_{\mathsf T}\ket v_{\mathsf{AS}}
 &\longmapsto
 -is\ket0_{\mathsf{PTA}}\bigl(\iota_0^\dagger\ket v\bigr)
 +i\ket1_{\mathsf P}\ket j_{\mathsf T}
 \bigl[I-(1-c)\Pi_{\mathsf A,0}\bigr]\ket v,\\
 \ket1_{\mathsf P}\ket k_{\mathsf T}\ket v_{\mathsf{AS}}
 &\longmapsto
 \ket1_{\mathsf P}\ket k_{\mathsf T}\ket v,
 \qquad k\ne j.
 \end{aligned}
 \label{eq:globalUpdate}
\end{equation}
Intuitively, apart from \(O(\log J)\) gates that relabel the time
register, \(R_j\) consists of a constant number of one-qubit gates on
\(\mathsf P\) controlled by the \(a+\log J\) qubits in \(\mathsf{AT}\).
It therefore requires \(O(a+\log J)\) one- and
two-qubit gates.  A reuse-register control adds only polylogarithmic
overhead in \(q\), and the action is trivial on \(\mathsf S\).
These rules define a unitary on the encoded public and private spaces;
we extend it by the identity on the unused register states.  In abstract
block notation, define
\begin{equation}
 S^\circ=R_{J-1}\cdots R_0,\qquad
 S=S^\circ(I\oplus\HAMT).
 \label{eq:transducerFactorization}
\end{equation}
Under Convention~\ref{conv:publicPrivate}, \(I\oplus\HAMT\) acts as
the identity on the public summand, applies \(\HAMT\) to registers
\(\mathsf{TAS}\) when \(\mathsf P=1\), and acts as the identity when
\(\mathsf P=0\).

As in \eqref{eq:catalystRegister}, define
\[
 \Gamma\ket{\psi_0}
 :=\sum_{j=0}^{J-1}\ket j_{\mathsf T}\ket{x_j},
\]
which gives a linear map \(\Gamma\colon\Sys\to\Priv\).

\begin{proposition}[One-query Cayley transducer]
\label{prop:globalTransducer}
The circuit \(S\) uses one \(\HAMT\) query, and
\(\norm{\Gamma}^2\leq\alpha T\).  For every
\(\ket{\psi_0}\in\Sys\),
\begin{equation}
 S\bigl(\ket{\psi_0}\oplus\Gamma\ket{\psi_0}\bigr)
 =U_C\ket{\psi_0}\oplus\Gamma\ket{\psi_0}.
 \label{eq:globalgraph}
\end{equation}
\end{proposition}

\begin{proof}
The single query acts on the catalyst components as specified in
\eqref{eq:fineOracle}, and hence
\[
 (I\oplus\HAMT)
 \left(\ket{\psi_0}\oplus\Gamma\ket{\psi_0}\right)
 =\ket{\psi_0}\oplus
 \sum_{j=0}^{J-1}\ket j_{\mathsf T}
 O_j\ket{x_j}.
\]
After the updates \(R_0,\ldots,R_{\ell-1}\), where
\(0\leq\ell\leq J\), the abstract state is
\[
 \ket{\psi_\ell}\oplus\left(
 \sum_{j=0}^{\ell-1}\ket j_{\mathsf T}\ket{x_j}
 +\sum_{j=\ell}^{J-1}\ket j_{\mathsf T}
 O_j\ket{x_j}\right).
\]
For \(\ell=0\), this expression is exactly the state produced by the
query.  Assume that the displayed expression holds after \(\ell\)
updates.  Equation~\eqref{eq:stepgraph} shows that
\(R_\ell\) replaces
\(\ket{\psi_\ell}\oplus O_\ell\ket{x_\ell}\) by
\(\ket{\psi_{\ell+1}}\oplus\ket{x_\ell}\), while the third rule in
\eqref{eq:globalUpdate} leaves every other time-index label unchanged.  This
proves the expression for \(\ell+1\).  At \(\ell=J\), it becomes
\(U_C\ket{\psi_0}\oplus\Gamma\ket{\psi_0}\), which proves
\eqref{eq:globalgraph}.  Since each \(R_j\) is independent of the
oracle, the factorization \eqref{eq:transducerFactorization} uses
exactly one \(\HAMT\) query.

For the catalyst bound, every Cayley factor is unitary, so
\(\norm{\ket{\psi_j}}=\norm{\ket{\psi_0}}\).  The definition of
\(\Gamma\) above, the orthogonality of the time-index basis, and
\eqref{eq:stepcost} give
\begin{equation}
 \begin{aligned}
 \norm{\Gamma\ket{\psi_0}}^2
 &=\sum_{j=0}^{J-1}\norm{\ket{x_j}}^2
 \leq w\sum_{j=0}^{J-1}\norm{\ket{\psi_j}}^2\\
 &=Jw\norm{\ket{\psi_0}}^2
 =\alpha T\norm{\ket{\psi_0}}^2.
 \end{aligned}
 \label{eq:W}
\end{equation}
\end{proof}

\section{Removing the catalyst with a finite reuse register}
\label{sec:reuse}

The transducer identity \eqref{eq:globalgraph} assumes that the catalyst
\(\Gamma\ket\psi\) is supplied.  Since this catalyst is generally
difficult to prepare, the identity does not directly yield a circuit
for \(U_C\).  We therefore use the reuse construction of Belovs,
Jeffery, and Yolcu~\cite[Theorem~3.2]{BJY24}, which makes \(N\) transducer calls
starting from zero private input.  This section defines the operator
\(P_N\) induced on the public space and derives the change
\(U_C-P_N\) in the public output caused by omitting the catalyst.  The
following sections exploit the structure of the transducer to obtain
faster error decay.

Write the public--private block decomposition of \(S\) as
\begin{equation}
 S=
 \begin{pmatrix}D&C\\ B&A\end{pmatrix}.
 \label{eq:transducerBlocks}
\end{equation}
Here
\[
 D\colon\Sys\to\Sys,\qquad C\colon\Priv\to\Sys,\qquad
 B\colon\Sys\to\Priv,\qquad A\colon\Priv\to\Priv.
\]
Because \(S\) is unitary, its blocks \(A\) and \(C\) are
contractions.
The transducer identity \eqref{eq:globalgraph} is equivalent to
\begin{equation}
 \Gamma=B+A\Gamma,\qquad U_C=D+C\Gamma.
 \label{eq:formalGraph}
\end{equation}

For an integer \(N\geq1\), let \(\mathsf K\) be a
\(\lceil\log_2N\rceil\)-qubit reuse register.  We use its first \(N\)
basis states.  The operations on \(\mathsf K\) below preserve their
span and are extended by the identity on the remaining labels.  This
register is distinct from the time-index
register \(\mathsf T\).  Its
value \(\ell\) identifies reuse call \(\ell\); it
does not represent physical time.  Choose a unitary \(F_N\) on
\(\mathsf K\) satisfying
\[
 F_N\ket0_{\mathsf K}
 =\frac1{\sqrt N}\sum_{\ell=0}^{N-1}\ket{\ell}_{\mathsf K}.
\]
Define the cyclic shift on \(\mathsf K\) by
\[
 X_N\ket{\ell}_{\mathsf K}
 =\ket{\ell+1\bmod N}_{\mathsf K}.
\]
Define
\[
 \widetilde F_N
 =F_N\otimes\ketbra{0}{0}_{\mathsf P}
  +I_{\mathsf K}\otimes\ketbra{1}{1}_{\mathsf P},
 \qquad
 \widetilde X_N
 =I_{\mathsf K}\otimes\ketbra{0}{0}_{\mathsf P}
  +X_N\otimes\ketbra{1}{1}_{\mathsf P}.
\]
The oracle call controlled by \(\mathsf P\) is
\[
 \operatorname{ctrl}(\HAMT)
 :=\ketbra{0}{0}_{\mathsf P}\otimes I_{\mathsf{TAS}}
  +\ketbra{1}{1}_{\mathsf P}\otimes\HAMT.
\]
With \(S^\circ\) defined in \eqref{eq:transducerFactorization}, for each
\(\ell\) define the unitary
\[
 S_\ell^\circ
 =\ketbra{\ell}{\ell}_{\mathsf K}\otimes S^\circ
 +(I_{\mathsf K}-\ketbra{\ell}{\ell}_{\mathsf K})
  \otimes I_{\mathsf{PTAS}}.
\]
Starting with \(\mathsf K=0\) and \(\mathsf P=0\), the circuit performs:
\begin{enumerate}
\item Apply \(\widetilde F_N\).
\item For each \(\ell=0,1,\ldots,N-1\), apply
\(\operatorname{ctrl}(\HAMT)\), \(S_\ell^\circ\), and
\(\widetilde X_N\), in that order.
\item Apply \(\widetilde F_N^\dagger\).
\end{enumerate}
Write the three operations in reuse call \(\ell\) as
\begin{equation}
 T_\ell:=\widetilde X_NS_\ell^\circ
 \bigl(I_{\mathsf K}\otimes\operatorname{ctrl}(\HAMT)\bigr),\qquad
 V_N:=\widetilde F_N^\dagger T_{N-1}\cdots T_0\widetilde F_N.
 \label{eq:VN}
\end{equation}
Each \(T_\ell\) contains the single query
\(I_{\mathsf K}\otimes\operatorname{ctrl}(\HAMT)\), which applies
\(\HAMT\) to the \(\mathsf P=1\) component independently of the value
of \(\mathsf K\), whereas only \(S_\ell^\circ\) is conditioned on
\(\mathsf K=\ell\).

Define \(P_N\colon\Sys\to\Sys\) by
\begin{equation}
 (\bra0_{\mathsf{KPTA}}\otimes I_{\mathsf S})
 V_N\bigl(\ket0_{\mathsf{KPTA}}\ket\psi_{\mathsf S}\bigr)
 =P_N\ket\psi_{\mathsf S}.
 \label{eq:VNpublicOutput}
\end{equation}
Equation~\eqref{eq:VNpublicOutput} is precisely the block-encoding
condition: \(V_N\) block-encodes \(P_N\), with \(\mathsf{KPTA}\) as
its auxiliary registers.
Let \(G_N(z):=N^{-1}\sum_{b=0}^{N-1}z^b\).
The following lemma expresses \(U_C-P_N\) in terms of the
private-to-private and private-to-public blocks of \(S\).

\begin{lemma}[Reuse error identity]
\label{lem:zeroCatalystBlock}
The operator \(P_N\) satisfies
\begin{equation}
 U_C-P_N=C\,G_N(A)\Gamma.
 \label{eq:genericResidual}
\end{equation}
\end{lemma}

\begin{proof}
Fix an input register state \(\ket\psi_{\mathsf S}\).  The first step applies
\(\widetilde F_N\) and gives
\begin{equation}
 \ket0_{\mathsf{KPTA}}\ket\psi
 \longrightarrow
 \frac1{\sqrt N}\sum_{r=0}^{N-1}
 \ket r_{\mathsf K}\ket0_{\mathsf{PTA}}\ket\psi.
 \label{eq:clockInitialState}
\end{equation}
Write \(\ket{z_\ell}_{\mathsf{TAS}}\in\Priv\) for the
unnormalized private vector entering call \(\ell\), with
\(\ket{z_0}=0\).  We prove by induction that, immediately before call
\(\ell\), the private component is
\(\ket\ell_{\mathsf K}\ket1_{\mathsf P}\ket{z_\ell}\); the public
component with label \(r<\ell\) contains
\(D\ket\psi/\sqrt N+C\ket{z_r}\), while every component with
\(r\geq\ell\) still contains \(\ket\psi/\sqrt N\).  The claim holds
for \(\ell=0\) by \eqref{eq:clockInitialState}.  Under this induction
hypothesis, the component with \(\mathsf K=\ell\) is
\[
 \ket\ell_{\mathsf K}\left[
 \frac1{\sqrt N}\ket0_{\mathsf{PTA}}\ket\psi
 +\ket1_{\mathsf P}\ket{z_\ell}
 \right].
\]
For this value of \(\mathsf K\), Convention~\ref{conv:publicPrivate}
identifies the register state with
\(\ket\psi/\sqrt N\oplus\ket{z_\ell}\), whose two summands have flags
\(\mathsf P=0\) and \(\mathsf P=1\), respectively.  The controlled
query and \(S_\ell^\circ\) both preserve \(\mathsf K\).  On the
component with \(\mathsf K=\ell\), the query applies \(\HAMT\) to the
private summand, after which \(S_\ell^\circ\) applies \(S^\circ\).
Suppressing the unchanged register state \(\ket\ell_{\mathsf K}\),
their joint action is
\begin{equation}
 \begin{aligned}
 S^\circ\left(
 \frac{\ket\psi}{\sqrt N}\oplus\HAMT\ket{z_\ell}
 \right)
 &=S\left(
 \frac{\ket\psi}{\sqrt N}\oplus\ket{z_\ell}
 \right)\\
 &=\left(\frac{D\ket\psi}{\sqrt N}+C\ket{z_\ell}\right)
 \oplus
 \left(\frac{B\ket\psi}{\sqrt N}+A\ket{z_\ell}\right).
 \end{aligned}
 \label{eq:reuseCallAction}
\end{equation}
The first equality is \eqref{eq:transducerFactorization}, and the
second uses the block decomposition \eqref{eq:transducerBlocks}.  Set
\begin{equation}
 \ket{z_{\ell+1}}
 :=\frac{B\ket\psi}{\sqrt N}
 +A\ket{z_\ell}.
 \label{eq:clockPrivateRecurrence}
\end{equation}
The public part of \eqref{eq:reuseCallAction} retains the label
\(\mathsf K=\ell\).  The final operation in call \(\ell\),
\(\widetilde X_N\), moves the private output to the next label:
\[
 \ket\ell_{\mathsf K}\ket1_{\mathsf P}
 \ket{z_{\ell+1}}
 \longrightarrow
 \ket{\ell+1\bmod N}_{\mathsf K}\ket1_{\mathsf P}
 \ket{z_{\ell+1}}.
\]
All other public components are unchanged during this call: the
controlled query and \(\widetilde X_N\) act trivially when
\(\mathsf P=0\), and \(S_\ell^\circ\) acts as the identity when
\(\mathsf K\ne\ell\).  This proves the induction step.  After all
\(N\) calls, the full state is
\begin{equation}
 \sum_{\ell=0}^{N-1}
 \ket\ell_{\mathsf K}\ket0_{\mathsf{PTA}}
 \left(\frac{D\ket\psi}{\sqrt N}
       +C\ket{z_\ell}\right)
 +\ket0_{\mathsf K}\ket1_{\mathsf P}\ket{z_N}.
 \label{eq:clockStateAfterLoop}
\end{equation}
Iterating \eqref{eq:clockPrivateRecurrence} and substituting
\(B=(I-A)\Gamma\) from the first identity in
\eqref{eq:formalGraph} gives
\begin{equation}
 \ket{z_\ell}
 =\frac1{\sqrt N}\sum_{k=0}^{\ell-1}A^kB\ket\psi
 =\frac{(I-A^\ell)\Gamma\ket\psi}{\sqrt N}.
 \label{eq:clockPrivateClosedForm}
\end{equation}

Applying \(\widetilde F_N^\dagger\) to
\eqref{eq:clockStateAfterLoop}, projecting onto \(\mathsf K=0\) and
\(\mathsf P=0\), and substituting \eqref{eq:clockPrivateClosedForm}
gives
\[
 \begin{aligned}
 P_N\ket\psi
 &=\frac1{\sqrt N}\sum_{\ell=0}^{N-1}
 \left(\frac{D\ket\psi}{\sqrt N}
       +C\ket{z_\ell}\right)\\
 &=D\ket\psi
 +\frac1N\sum_{\ell=0}^{N-1}
 \sum_{k=0}^{\ell-1}CA^kB\ket\psi.
 \end{aligned}
\]
For a fixed \(k\), the inner term occurs for
\(\ell=k+1,\ldots,N-1\), a total of \(N-1-k\) times.  Hence
\begin{equation}
 P_N=D+\sum_{k=0}^{N-2}
 \left(1-\frac{k+1}{N}\right)CA^kB.
 \label{eq:PN}
\end{equation}
For \(N=1\), the sum is empty and \(P_1=D\).
Using the first identity in \eqref{eq:formalGraph}, a weighted
telescoping sum gives
\[
 \begin{aligned}
 \sum_{k=0}^{N-2}\left(1-\frac{k+1}{N}\right)A^kB
 &=\sum_{k=0}^{N-2}\left(1-\frac{k+1}{N}\right)
   (A^k-A^{k+1})\Gamma\\
 &=\left(I-\frac1N\sum_{k=0}^{N-1}A^k\right)\Gamma\\
 &=\bigl(I-G_N(A)\bigr)\Gamma.
 \end{aligned}
\]
Substituting this identity into \eqref{eq:PN} and using
the second identity in \eqref{eq:formalGraph} proves
\eqref{eq:genericResidual}.
\end{proof}

\section{Time ordering and factorial decay}

Lemma~\ref{lem:zeroCatalystBlock} gives the error caused by starting
with zero private input:
\[
 C\,G_N(A)\Gamma,\qquad
 G_N(A)=\frac1N\sum_{k=0}^{N-1}A^k.
\]
A direct estimate of \(\norm{G_N(A)}\) gives a bound that decreases too
slowly with \(N\) to obtain the desired dependence on precision.  We
instead combine several reuse lengths.  To choose that combination,
we first determine the structure of \(A\).

\subsection{Lower-triangular structure of the private block}

For an operator \(X\) on \(\Priv\), write
\begin{equation}
 [X]_{r,j}:=
 (\bra r_{\mathsf T}\otimes I)X
 (\ket j_{\mathsf T}\otimes I)
 \colon\Anc\otimes\Sys\longrightarrow\Anc\otimes\Sys .
\label{eq:timeBlockDefinition}
\end{equation}
Define \(\Delta\) by
\begin{align}
 [\Delta]_{j,j}&=(1-c)\Pi_{\mathsf A,0}, \label{eq:Ddiag}\\
 [\Delta]_{r,j}&=
 s^2c^{r-j-1}\Pi_{\mathsf A,0}\quad(r>j),&
 [\Delta]_{r,j}&=0\quad(r<j).
 \label{eq:Doff}
\end{align}
Here \(c\) and \(s\) are defined in \eqref{eq:cayleyUpdate}, and
\(\Pi_{\mathsf A,0}\) projects \(\mathsf A\) onto
\(\ket0_{\mathsf A}\).
Thus \(\Delta\) never sends an input with label \(j\) to an earlier
label \(r<j\); this is the time ordering used below.

\begin{lemma}[Lower-triangular form of the private block]
\label{lem:chronologicalBlock}
The private block \(A\) of \(S\) satisfies
\begin{equation}
 A=i(I-\Delta)\HAMT.
 \label{eq:Aform}
\end{equation}
\end{lemma}

\begin{proof}
Let \(A^\circ\) be the private block of \(S^\circ\) defined in
\eqref{eq:transducerFactorization}.  We compute its blocks in the
time-index basis using \eqref{eq:timeBlockDefinition}.  The product
\(S^\circ=R_{J-1}\cdots R_0\) applies
\(R_0,R_1,\ldots,R_{J-1}\) from first to last.  The third rule in
\eqref{eq:globalUpdate} shows that
\(R_0,\ldots,R_{j-1}\) act as the identity on the private component
whose time-register label is \(j\).  The private-to-private block of
\(R_j\) on this component is \(i[I-(1-c)\Pi_{\mathsf A,0}]\), and later
updates do not act on the label \(j\).
Consequently, the diagonal block is
\[
 [A^\circ]_{j,j}=i\bigl[I-(1-c)\Pi_{\mathsf A,0}\bigr].
\]

The update order also determines every off-diagonal block.  A private
input with label \(j\) is unchanged until \(R_j\).  To finish with a
different private label \(r\), it must enter the public space at
\(R_j\), remain there until \(R_r\), and return to the private space at
\(R_r\).  Since the updates are applied in increasing order, this is
impossible when \(r<j\).  For \(r>j\) and
\(\ket v\in\Anc\otimes\Sys\), the unique path is
\[
\begin{aligned}
 \ket1_{\mathsf P}\ket j_{\mathsf T}\ket v_{\mathsf{AS}}
 &\longmapsto
 -is\ket0_{\mathsf{PTA}}
 \bigl(\iota_0^\dagger\ket v\bigr)\\
 &\longmapsto
 -is c^{r-j-1}\ket0_{\mathsf{PTA}}
 \bigl(\iota_0^\dagger\ket v\bigr)\\
 &\longmapsto
 -is^2c^{r-j-1}
 \ket1_{\mathsf P}\ket r_{\mathsf T}
 \Pi_{\mathsf A,0}\ket v.
\end{aligned}
\]
The three arrows apply the private-to-public block at \(R_j\), the
public-to-public block \(cI_{\Sys}\) at the \(r-j-1\) intermediate
updates, and the public-to-private block at \(R_r\), respectively.
Combining these off-diagonal blocks with the diagonal block gives
\begin{equation}
 [A^\circ]_{r,j}
 =
 \begin{cases}
 0,&r<j,\\
 i\bigl[I-(1-c)\Pi_{\mathsf A,0}\bigr],&r=j,\\
 -is^2c^{r-j-1}\Pi_{\mathsf A,0},&r>j.
 \end{cases}.
\end{equation}
Comparing this formula with \eqref{eq:Ddiag}--\eqref{eq:Doff} gives
\(A^\circ=i(I-\Delta)\), and
\eqref{eq:transducerFactorization} then implies
\(A=A^\circ\HAMT=i(I-\Delta)\HAMT\).
\end{proof}

\begin{samepage}
For finitely many real coefficients satisfying
\(\sum_N\lambda_N=1\), define
\[
 \widetilde U:=\sum_N\lambda_NP_N,\qquad
 E(z):=\sum_N\lambda_NG_N(z).
\]
Since each \(V_N\) block-encodes \(P_N\), the LCU construction for
block-encoded matrices~\cite[Lemma~52]{GSLW19} produces a
block-encoding of \(\widetilde U\) with normalization
\(\sum_N|\lambda_N|\), once the coefficient state is prepared and a
unitary applies \(\operatorname{sgn}(\lambda_N)V_N\) conditioned on the
selected value of \(N\).  These operations and the subsequent amplification
are constructed in
Section~\ref{subsec:lcuCircuit}.  Summing
\eqref{eq:genericResidual} and using \(\sum_N\lambda_N=1\) gives the
first line below.  In the block decomposition
\eqref{eq:transducerBlocks}, \(C\) is a block of the unitary \(S\), so
\(\norm C\leq1\), while \eqref{eq:W} gives
\(\norm\Gamma\leq\sqrt{\alpha T}\); these inequalities give the
second line:
\begin{equation}
 \begin{aligned}
 U_C-\widetilde U&=C E(A)\Gamma,\\
 \big\|U_C-\widetilde U\big\|
 &\leq\norm C\,\norm{E(A)}\,\norm\Gamma
 \leq\sqrt{\alpha T}\,\norm{E(A)}.
 \end{aligned}
 \label{eq:weightedResidual}
\end{equation}
It remains to find a polynomial \(E\), expressible as such a weighted
sum, for which \(\norm{E(A)}\) is small.
\end{samepage}

Equation~\eqref{eq:Aform} suggests a polynomial for which time
ordering gives a useful bound.  Using \(\HAMT^2=I\),
\[
 A^2=-I+\Delta+\HAMT\Delta\HAMT
       -\Delta\HAMT\Delta\HAMT.
\]
The first term is the only one containing no \(\Delta\), and adding
\(I\) cancels it.  Normalizing at \(z=1\), as required by
\(\sum_N\lambda_N=1\), define
\begin{equation}
 Q(z):=\frac{1+z^2}{2}.
 \label{eq:Qpoly}
\end{equation}
Substituting \(A\) into \eqref{eq:Qpoly} gives
\begin{equation}
 Q(A)
 =\frac12\bigl(
 \Delta+\HAMT\Delta\HAMT-\Delta\HAMT\Delta\HAMT
 \bigr).
 \label{eq:Q}
\end{equation}
Every term on the right contains at least one \(\Delta\).  Hence every
term in the expansion of \(Q(A)^q\) contains at least \(q\) such
factors.  Since \(\Delta\) never decreases the time label, the possible
sequences of intermediate time labels are ordered.

\subsection{Factorial bound}

We now sum over these ordered sequences to bound \(\norm{Q(A)^q}\)
uniformly in the number of Cayley steps.

\begin{proposition}[Factorial decay]
\label{prop:factorialDecay}
If \(q\geq\max\{1,\alpha T\}\) is an integer and
\(w\leq\alpha T/q\), then
\begin{equation}
 \norm{Q(A)^q}\leq\left(\frac{12e\alpha T}{q}\right)^q.
 \label{eq:factorial}
\end{equation}
\end{proposition}

\begin{proof}
The assumptions imply \(0<w\leq1\), and hence
\[
 1-c=\frac{w}{1+w/2}\leq w,\qquad
 s\leq\sqrt{2w},\qquad
 c=\frac{1-w/2}{1+w/2}\leq e^{-w}.
\]
From \eqref{eq:Ddiag}--\eqref{eq:Doff},
\[
 \norm{[\Delta]_{j,j}}\leq w,\qquad
 \norm{[\Delta]_{r,j}}
 \leq2w e^{-(r-j-1)w}
 \quad(r>j).
\]
Equations \eqref{eq:timeBlockDefinition} and \eqref{eq:fineOracle} give
\begin{equation}
 [\HAMT]_{r,j}=\delta_{r,j}O_j.
 \label{eq:HAMTTimeBlocks}
\end{equation}
Expanding the block product with \eqref{eq:HAMTTimeBlocks}, and using
the lower-triangular structure in
\eqref{eq:Ddiag}--\eqref{eq:Doff}, gives
\[
 [\Delta\HAMT\Delta\HAMT]_{r,j}
 =\sum_{k=j}^{r}
 [\Delta]_{r,k}O_k[\Delta]_{k,j}O_j.
\]
Each \(O_k\) is unitary by Definition~\ref{def:HAMTOracle}, so
\[
 \norm{[\Delta\HAMT\Delta\HAMT]_{r,j}}
 \leq\sum_{k=j}^{r}
 \norm{[\Delta]_{r,k}}\norm{[\Delta]_{k,j}}.
\]
For \(r>j\), the endpoint terms \(k=j,r\) contribute at most \(4w\).
Writing \(d=r-j-1\), the remaining terms satisfy
\[
 \sum_{j<k<r}
 \norm{[\Delta]_{r,k}}\norm{[\Delta]_{k,j}}
 \leq4w^2d e^{-(d-1)w}\leq4w.
\]
Here \(dw e^{-(d-1)w}\leq1\) for \(0<w\leq1\).  For \(r=j\), the same
bound follows directly from \eqref{eq:Ddiag}, so
\[
 \sum_{k=j}^{r}
 \norm{[\Delta]_{r,k}}\norm{[\Delta]_{k,j}}
 \leq8w.
\]
In \eqref{eq:Q}, each of the first two terms is bounded by \(2w\), and
the third is bounded by the convolution above.  Thus
\begin{equation}
 \norm{[Q(A)]_{r,j}}
 \leq\frac12(2w+2w+8w)=6w\quad(r\geq j),\qquad
 [Q(A)]_{r,j}=0\quad(r<j).
 \label{eq:Qdom}
\end{equation}
The lower-triangular bound in \eqref{eq:Qdom} gives
\[
 [Q(A)^q]_{r,j}
 =\sum_{j\leq i_1\leq\cdots\leq i_{q-1}\leq r}
 [Q(A)]_{r,i_{q-1}}\cdots[Q(A)]_{i_1,j}.
\]
Applying \eqref{eq:Qdom} to each factor yields
\[
 \norm{[Q(A)^q]_{r,j}}
 \leq(6w)^q\binom{r-j+q-1}{q-1}.
\]
Schur's test applied to the block matrix gives
\[
 \norm{X}\leq
 \left[
 \left(\max_r\sum_j\norm{[X]_{r,j}}\right)
 \left(\max_j\sum_r\norm{[X]_{r,j}}\right)
 \right]^{1/2}.
\]
The identity
\[
 \sum_{r=0}^{R}\binom{r+q-1}{q-1}=\binom{R+q}{q},
\]
shows that every row and column sum is at most
\((6w)^q\binom{J+q-1}{q}\).  Therefore
\begin{equation}
 \norm{Q(A)^q}\leq
 (6w)^q\binom{J+q-1}{q}.
 \label{eq:Schur}
\end{equation}
Since \(Jw=\alpha T\) and
\(\binom{J+q-1}{q}\leq(e(J+q-1)/q)^q\), the right-hand side of
\eqref{eq:Schur} is at most
\[
 \left(\frac{6e(\alpha T+qw)}{q}\right)^q
 \leq\left(\frac{12e\alpha T}{q}\right)^q,
\]
where the last inequality uses \(w\leq\alpha T/q\).
\end{proof}

The right-hand side of \eqref{eq:factorial} is independent of \(J\).
Since \(w=\alpha T/J\), the condition \(w\leq\alpha T/q\) is equivalent
to \(J\geq q\).  The estimate therefore holds uniformly for every
admissible Cayley step count \(J\geq q\).

\section{Combining the reuse circuits}
\label{sec:combining}

The preceding section bounds the factor \(Q(A)^q\) used below.  We now
choose the coefficients of the reuse circuits, construct the resulting
block-encoding, and select the parameters used in Theorem~\ref{thm:main}.

\subsection{A weighted combination of reuse lengths}

Equation~\eqref{eq:weightedResidual} shows that we need a polynomial
\(E_q\) for which \(\norm{E_q(A)}\) is small and
\(\sum_N|\lambda_N|\), the LCU normalization of
\(E_q=\sum_N\lambda_NG_N\), remains bounded.  Guided by
Proposition~\ref{prop:factorialDecay}, we choose \(E_q\) to contain the
factor \(Q^q\).  Fix an integer
\(q\geq\max\{1,\alpha T\}\) and suppose \(w\leq\alpha T/q\), where
\(w=\alpha T/J\) as in Section~\ref{subsec:cayley}.

We first record the change of basis.  Let
\(G_N(z)=N^{-1}\sum_{k=0}^{N-1}z^k\) and
\(E(z)=\sum_{k=0}^{L-1}e_kz^k\), and set \(e_L=0\).  Then
\begin{equation}
 E(z)=\sum_{N=1}^L\lambda_NG_N(z),
 \qquad
 \lambda_N=N(e_{N-1}-e_N).
 \label{eq:clockExpansion}
\end{equation}
Indeed, the coefficient of \(z^k\) on the right is
\[
 \sum_{N=k+1}^{L}\frac{\lambda_N}{N}
 =\sum_{N=k+1}^{L}(e_{N-1}-e_N)
 =e_k-e_L=e_k.
\]
Evaluating \eqref{eq:clockExpansion} at \(z=1\) and using
\(G_N(1)=1\) also gives \(\sum_N\lambda_N=E(1)\).  The corresponding
LCU normalization is
\[
 \sum_N|\lambda_N|=\sum_N N|e_{N-1}-e_N|,
\]
which is the weighted variation between neighboring monomial
coefficients of \(E\).

The monomial coefficients of \(Q(z)^q\) alternate between zero and a
binomial coefficient.  Multiplication by \(G_{2q}\) averages
\(2q\) consecutive coefficients, reducing the differences that enter
the normalization above.  We therefore define

\begin{equation}
 E_q(z):=Q(z)^qG_{2q}(z)
 =\left(\frac{1+z^2}{2}\right)^q
 \left(\frac1{2q}\sum_{\ell=0}^{2q-1}z^\ell\right)
 =\sum_{k=0}^{4q-1}e_kz^k.
 \label{eq:Eq}
\end{equation}
The numbers
\[
 b_j:=2^{-q}\binom{q}{j},\qquad 0\leq j\leq q,
\]
are the monomial coefficients in
\(Q(z)^q=\sum_{j=0}^qb_jz^{2j}\); set \(b_j=0\) outside this range.
Multiplying the two factors in \eqref{eq:Eq} gives
\[
 E_q(z)=\frac1{2q}
 \sum_{j=0}^{q}\sum_{\ell=0}^{2q-1}
 b_jz^{2j+\ell}.
\]
Consequently, for \(0\leq j\leq2q-1\),
\begin{equation}
 e_{2j}=e_{2j+1}
 =\frac1{2q}\sum_{\ell=0}^{q-1}b_{j-\ell}.
 \label{eq:pairedCoefficients}
\end{equation}
Set \(e_{4q}=0\) and apply
\eqref{eq:clockExpansion} with \(L=4q\) to define
\(\lambda_1,\ldots,\lambda_{4q}\).  We also set
\[
 \Lambda:=\sum_{N=1}^{4q}|\lambda_N|,\qquad
 \widetilde U:=\sum_{N=1}^{4q}\lambda_NP_N.
\]
We now show that this choice keeps \(\Lambda<2\) while preserving the
factorial error bound.

\begin{proposition}[Combination of reuse lengths]
\label{prop:clockFilter}
The coefficients satisfy
\(\sum_{N=1}^{4q}\lambda_N=1\) and
\(\Lambda=2-2^{1-q}<2\).
The approximation error obeys
\begin{equation}
 \big\|U_C-\widetilde U\big\|\leq\sqrt{\alpha T}
 \left(\frac{12e\alpha T}{q}\right)^q.
 \label{eq:filterError}
\end{equation}
\end{proposition}

\begin{proof}
Evaluating \eqref{eq:clockExpansion} at \(z=1\) gives
\(\sum_N\lambda_N=E_q(1)=Q(1)^qG_{2q}(1)=1\).
Substituting \(E_q\) into \eqref{eq:weightedResidual} gives
\begin{equation}
 U_C-\widetilde U=CE_q(A)\Gamma
 =CQ(A)^qG_{2q}(A)\Gamma.
 \label{eq:Pfilt}
\end{equation}
Since \(A\) and \(C\) are blocks of the unitary \(S\) in
\eqref{eq:transducerBlocks}, both are contractions.  Hence
\[
 \norm{G_{2q}(A)}
 \leq\frac1{2q}\sum_{\ell=0}^{2q-1}\norm A^\ell
 \leq1.
\]
Together with \(\norm C\leq1\), \eqref{eq:W}, and
\eqref{eq:factorial}, equation~\eqref{eq:Pfilt} gives
\[
 \big\|U_C-\widetilde U\big\|
 \leq\norm C\,\norm{Q(A)^q}\,\norm{G_{2q}(A)}\,\norm\Gamma
 \leq\sqrt{\alpha T}
 \left(\frac{12e\alpha T}{q}\right)^q,
\]
which is \eqref{eq:filterError}.

To compute \(\Lambda\), substitute
\eqref{eq:pairedCoefficients} into the coefficient formula in
\eqref{eq:clockExpansion}.  For \(0\leq j\leq2q-1\),
\begin{align*}
 \lambda_{2j+1}
 &=(2j+1)(e_{2j}-e_{2j+1})=0,\\
 \lambda_{2j+2}
 &=(2j+2)(e_{2j+1}-e_{2j+2})\\
 &=\frac{j+1}{q}
   \left(
    \sum_{\ell=0}^{q-1}b_{j-\ell}
    -\sum_{\ell=0}^{q-1}b_{j+1-\ell}
   \right)\\
 &=\frac{j+1}{q}\bigl(b_{j-q+1}-b_{j+1}\bigr).
\end{align*}
Because \(b_k=0\) outside \(0\leq k\leq q\) and
\(b_0=b_q=2^{-q}\), the formula for \(\lambda_{2j+2}\) becomes
\[
 \lambda_{2j+2}
 =
 \begin{cases}
  -\dfrac{j+1}{q}b_{j+1},&0\leq j\leq q-2,\\[4pt]
  0,&j=q-1,\\[4pt]
  \dfrac{j+1}{q}b_{j-q+1},&q\leq j\leq2q-1.
 \end{cases}
\]
The piecewise formula shows that the only negative coefficients are
\(\lambda_{2j+2}\) for \(0\leq j\leq q-2\).  Since all odd
coefficients vanish and \(\sum_N\lambda_N=1\),
\begin{align*}
 \Lambda=\sum_N\lambda_N
   -2\sum_{j=0}^{q-2}\lambda_{2j+2}=1+\frac2q\sum_{k=1}^{q-1}kb_k.
\end{align*}
Using \(k\binom qk=q\binom{q-1}{k-1}\), we obtain
\[
 \sum_{k=0}^qkb_k
 =2^{-q}\sum_{k=1}^q k\binom qk
 =\frac{q}{2^q}\sum_{k=1}^q\binom{q-1}{k-1}
 =\frac q2.
\]
Using this identity and \(b_q=2^{-q}\), we obtain
\[
 \Lambda
 =1+\frac2q\left(\sum_{k=0}^qkb_k-qb_q\right)
 =1+\frac2q\left(\frac q2-q2^{-q}\right)
 =2-2^{1-q}.
\]
\end{proof}

\subsection{Circuit implementation and amplification}
\label{subsec:lcuCircuit}

The preceding subsection gives coefficients \(\lambda_N\) with
\(\Lambda<2\).  We now construct a block-encoding of
\(\widetilde U=\sum_N\lambda_NP_N\) and amplify it.
Equation~\eqref{eq:VNpublicOutput} identifies each reuse unitary
\(V_N\) as a block-encoding of \(P_N\).  By the linear-combination
lemma for block-encoded matrices~\cite[Lemma~52]{GSLW19}, it
suffices to prepare the coefficients and construct a unitary that
applies \(\operatorname{sgn}(\lambda_N)V_N\) conditioned on the
selected value of \(N\).  We pad the coefficient sum to
\(2\) with a block-encoding of the zero operator.  This gives a
block-encoding of \(\widetilde U=\sum_N\lambda_NP_N\) with
normalization \(2\).  The construction below implements this
conditional unitary with \(4q\) oracle queries.

To represent every \(V_N\), for \(N\leq4q\), on the same registers, use
a \(\lceil\log_2(4q)\rceil\)-qubit reuse register \(\mathsf K\) and
extend its reuse-register operations by the identity on the labels
\(N,\ldots,4q-1\).  The span of
\(\ket0_{\mathsf K},\ldots,\ket{N-1}_{\mathsf K}\) is invariant, and
the restriction to this span is the original \(V_N\).  For a selected
value of \(N\), write \(T_\ell^{(N)}\) for the corresponding extension
of \(T_\ell\) in \eqref{eq:VN}; only its final shift
\(\widetilde X_N\) depends on \(N\).

Let the selection register \(\mathsf B\) contain the states
\(\ket0_{\mathsf B}\), \(\ket{N}_{\mathsf B}\) for
\(\lambda_N\neq0\), and an extra state \(\ket{\perp}_{\mathsf B}\).
Since \(\Lambda\leq2\), choose a unitary
\(\operatorname{PREP}\) satisfying
\[
 \operatorname{PREP}\ket0_{\mathsf B}
 =
 \sum_{\lambda_N\neq0}
 \sqrt{\frac{|\lambda_N|}{2}}\ket{N}_{\mathsf B}
 +\sqrt{\frac{2-\Lambda}{2}}\ket{\perp}_{\mathsf B}.
\]

Define the corresponding SELECT unitary by the following action.  For
every state
\(\ket\phi\) of
\(\mathsf{KPTAS}\),
\begin{equation}
 \operatorname{SELECT}
 \bigl(\ket N_{\mathsf B}\ket\phi\bigr)
 =\operatorname{sgn}(\lambda_N)
 \ket N_{\mathsf B}V_N\ket\phi.
 \label{eq:SELECTaction}
\end{equation}
Algorithm~\ref{alg:SELECT} implements \eqref{eq:SELECTaction} using the
reuse operations from \eqref{eq:VN}.  Comparisons with \(N\) return
false when \(\mathsf B\) is in \(\ket\perp_{\mathsf B}\) or an unused
basis state.

\begin{algorithm}[H]
\caption{Implementation of \(\operatorname{SELECT}\)}
\label{alg:SELECT}
\begin{algorithmic}[1]
\State For each \(N\) with \(\lambda_N\neq0\), controlled on
       \(\mathsf B=N\), apply the phase \(\operatorname{sgn}(\lambda_N)\)
       and then \(\widetilde F_N\).
\For{\(\ell=0,\ldots,4q-1\)}
  \State From \(\mathsf B\), set a work qubit to \(\ket1\) if
         \(N>\ell\), and to \(\ket0\) otherwise.
  \State Controlled on the work qubit, apply
         \(I_{\mathsf K}\otimes\operatorname{ctrl}(\HAMT)\), then
         \(S_\ell^\circ\), and then the shift \(\widetilde X_N\)
         selected by the value \(N\) in \(\mathsf B\).
  \State Uncompute the work qubit.
\EndFor
\State For each \(N\) with \(\lambda_N\neq0\), controlled on
       \(\mathsf B=N\), apply \(\widetilde F_N^\dagger\).
\State On \(\mathsf B=\perp\), flip \(\mathsf P\).
\end{algorithmic}
\end{algorithm}

For a selected \(N\), the loop in Algorithm~\ref{alg:SELECT} executes
\(T_\ell^{(N)}\) exactly for
\(\ell=0,\ldots,N-1\).  Apart from the controlled sign, this component
therefore undergoes
\(\widetilde F_N^\dagger T_{N-1}^{(N)}\cdots
T_0^{(N)}\widetilde F_N=V_N\) by
\eqref{eq:VN}.  Each of the \(4q\) layers contains one \(\HAMT\) query
controlled on \(N>\ell\) and \(\mathsf P=1\), so SELECT uses \(4q\)
oracle calls.

The last step maps the \(\ket\perp_{\mathsf B}\) component to
\(\mathsf P=1\), while SELECT is the identity on unused basis states of
\(\mathsf B\).  Reversing these steps gives
\(\operatorname{SELECT}^\dagger\), again with \(4q\) calls because
\(\HAMT^\dagger=\HAMT\).

Let \(\mathsf W\) collect all auxiliary registers, including
\(\mathsf B,\mathsf K,\mathsf P,\mathsf T,\mathsf A\), and
the temporary comparison controls, and
let \(\ket0_{\mathsf W}\) denote their joint all-zero state.  Define
\[
 V=(\operatorname{PREP}^\dagger\otimes I)
 \operatorname{SELECT}
 (\operatorname{PREP}\otimes I).
\]
For any operator \(Y\) on the work and system registers, write
\[
 [Y]_0
 :=(\bra0_{\mathsf W}\otimes I)Y
   (\ket0_{\mathsf W}\otimes I).
\]
The \(\ket\perp_{\mathsf B}\) term has \(\mathsf P=1\) after SELECT
and therefore does not contribute to \([V]_0\).  For every other
label \(N\), \(\operatorname{PREP}\) and
\(\operatorname{PREP}^\dagger\) each contribute a factor
\(\sqrt{|\lambda_N|/2}\).  Together with the controlled sign, this
gives the contribution \(\lambda_NP_N/2\) to \([V]_0\).
Equations~\eqref{eq:VNpublicOutput} and \eqref{eq:SELECTaction}
therefore give
\begin{equation}
 \begin{aligned}
 {}[V]_0
 &=\sum_{\lambda_N\neq0}
   \frac{|\lambda_N|}{2}\operatorname{sgn}(\lambda_N)P_N\\
 &=\frac12\sum_{\lambda_N\neq0}\lambda_NP_N
 =\frac{\widetilde U}{2}.
\end{aligned}
\label{eq:half}
\end{equation}
Thus, \(V\) is a
block-encoding of \(\widetilde U\) with normalization \(2\).

Let
\[
 \Pi_{\mathsf W,0}:=\ketbra{0}{0}_{\mathsf W}\otimes I,
 \qquad
 \mathcal R_{\mathsf W}:=2\Pi_{\mathsf W,0}-I .
\]
The unitary \(V\) may also produce components orthogonal to
\(\ket0_{\mathsf W}\).  The degree-three case of robust oblivious
amplitude amplification corrects the normalization and suppresses
these components~\cite{BerryEtAl15,GSLW19}.  In our reflection
convention, the degree-three amplification sequence of Berry
et al.~\cite[Eq.~(12)]{BerryEtAl15} is
\[
 U_{\mathrm{sim}}
 :=-V\mathcal R_{\mathsf W}V^\dagger
       \mathcal R_{\mathsf W}V.
\]
Their reflection \(I-2\Pi_{\mathsf W,0}\) equals
\(-\mathcal R_{\mathsf W}\).  The two minus signs cancel because the
sequence contains two reflections.  Taking
\(P=\Pi_{\mathsf W,0}\), \(W=V\), and \(s=2\) in
\cite[Eqs.~(13)--(14)]{BerryEtAl15}, and using \eqref{eq:half}, gives
\begin{equation}
 [U_{\mathrm{sim}}]_0
 =3[V]_0-4[V]_0[V]_0^\dagger[V]_0
 =\frac32\widetilde U
 -\frac12\widetilde U\widetilde U^\dagger\widetilde U.
\label{eq:OAA}
\end{equation}
The reflections act only on \(\mathsf W\), so the same circuit works
for every system input \(\ket\psi_{\mathsf S}\).
The following lemma converts the block identity \eqref{eq:OAA} into the
corresponding bound on the full output state.

\begin{lemma}[Robust one-step amplification]
\label{lem:robustOAA}
Let \(U\) be unitary and set
\(\eta=\norm{\widetilde U-U}\).  If \(\eta\leq1/8\), then, for every unit state
\(\ket\psi_{\mathsf S}\),
\begin{equation}
 \norm{
 U_{\mathrm{sim}}\bigl(\ket0_{\mathsf W}\ket\psi\bigr)
 -\ket0_{\mathsf W}U\ket\psi}
 \leq3\eta.
 \label{eq:stateerror}
\end{equation}
\end{lemma}

\begin{proof}
Use the polar decomposition \(\widetilde U=Z\Sigma\), where
\(\Sigma\geq0\).  Every singular value of \(\widetilde U\) lies in
\([1-\eta,1+\eta]\), so \(Z\) is unitary and
\(\norm{\Sigma-I}\leq\eta\).  Consequently,
\[
 \norm{Z-U}
 \leq\norm{Z-\widetilde U}+\norm{\widetilde U-U}
 =\norm{I-\Sigma}+\eta
 \leq2\eta.
\]
Equation \eqref{eq:OAA} gives
\[
 [U_{\mathrm{sim}}]_0=Zf(\Sigma),\qquad
 f(\sigma)=\frac{\sigma(3-\sigma^2)}2.
\]
Writing \(\sigma=1+t\), we have
\[
 f(1+t)-1=-\frac32t^2-\frac12t^3.
\]
Therefore, if \(|t|\leq\eta\leq1/8\), then
\[
 \lvert f(1+t)-1\rvert
 \leq\frac32\eta^2+\frac12\eta^3
 \leq\frac{13}{8}\eta^2.
\]
Applying this bound to every singular value of \(\widetilde U\) gives
\(\norm{f(\Sigma)-I}\leq13\eta^2/8\).  Hence
\[
 \begin{aligned}
 \norm{[U_{\mathrm{sim}}]_0-U}
 &=\norm{Zf(\Sigma)-U}\\
 &\leq\norm{f(\Sigma)-I}+\norm{Z-U}\\
 &\leq2\eta+\frac{13}{8}\eta^2.
 \end{aligned}
\]
The same singular-value bound gives, for every unit \(\ket\psi\),
\[
 \norm{f(\Sigma)\ket\psi}
 \geq1-\frac{13}{8}\eta^2.
\]
It follows that
\[
 1-\norm{f(\Sigma)\ket\psi}^2
 \leq1-\left(1-\frac{13}{8}\eta^2\right)^2
 \leq\frac{13}{4}\eta^2.
\]
Unitarity of \(U_{\mathrm{sim}}\) therefore gives
\[
 \norm{(I-\Pi_{\mathsf W,0})U_{\mathrm{sim}}
   \bigl(\ket0_{\mathsf W}\ket\psi\bigr)}^2
 =1-\norm{[U_{\mathrm{sim}}]_0\ket\psi}^2
 \leq\frac{13}{4}\eta^2.
\]
The error vector has the orthogonal decomposition
\[
 \begin{aligned}
 U_{\mathrm{sim}}\bigl(\ket0_{\mathsf W}\ket\psi\bigr)
   -\ket0_{\mathsf W}U\ket\psi=
 \ket0_{\mathsf W}\bigl([U_{\mathrm{sim}}]_0-U\bigr)\ket\psi
 +(I-\Pi_{\mathsf W,0})U_{\mathrm{sim}}
   \bigl(\ket0_{\mathsf W}\ket\psi\bigr).
 \end{aligned}
\]
The first summand lies in \(\operatorname{ran}\Pi_{\mathsf W,0}\),
whereas the second lies in \(\ker\Pi_{\mathsf W,0}\).  The Pythagorean
theorem therefore gives
\[
 \norm{
 U_{\mathrm{sim}}
 \bigl(\ket0_{\mathsf W}\ket\psi\bigr)
 -\ket0_{\mathsf W}U\ket\psi}^2
 \leq
 \left(2\eta+\frac{13}{8}\eta^2\right)^2
 {}+\frac{13}{4}\eta^2
 =\eta^2\left[\left(2+\frac{13}{8}\eta\right)^2
   +\frac{13}{4}\right]
 \leq9\eta^2,
\]
where the last inequality uses \(\eta\leq1/8\).  This proves
\eqref{eq:stateerror}.
\end{proof}

The circuit \(U_{\mathrm{sim}}\) uses \(V,V^\dagger,V\), each with
\(4q\) oracle calls.  Its total \(\HAMT\) query count is
\begin{equation}
 3(4q)=12q.
 \label{eq:querycount}
\end{equation}

\subsection{Query and gate complexity}

The identity circuit suffices when \(\alpha T\leq\eps\).  For
\(\alpha T>\eps\), the theorem below gives the complexity of the
construction, which is summarized in
Algorithm~\ref{alg:timeDependentSimulation}.

\begin{theorem}[Query and gate complexity]
\label{thm:main}
Let \(H\colon[0,T]\to\operatorname{Herm}(\Sys)\) satisfy
\(\norm{H(t)}\leq\alpha\) and
\(\norm{H(t)-H(s)}\leq\beta|t-s|\) for all \(s,t\in[0,T]\), and be given
through the \(\HAMT\) oracle in Definition~\ref{def:HAMTOracle}.  For every
\(0<\eps\leq1/2\), there exist an auxiliary register \(\mathsf W\) and
a unitary circuit \(U_{\mathrm{sim}}\) such that, for every unit initial
state \(\ket{\psi_0}_{\mathsf S}\),
\[
 \norm{
 U_{\mathrm{sim}}\bigl(\ket0_{\mathsf W}\otimes\ket{\psi_0}\bigr)
 -\ket0_{\mathsf W}\otimes U_H(T)\ket{\psi_0}}
 \leq\eps.
\]
The circuit uses \(12q\) queries, where
\[
 q=O\!\left(
 \alpha T+\frac{\log(1/\eps)}
 {\log\!\bigl(e+\log(1/\eps)/(\alpha T)\bigr)}
 \right)
\]
and can be implemented using
\(\widetilde O\bigl(q(a+1)\allowbreak
\max\{q,\allowbreak\beta T^2/\eps,\allowbreak
(\alpha T)^{3/2}/\sqrt\eps\}\bigr)\)
one- and two-qubit gates.
\end{theorem}

\begin{algorithm}[t]
\caption{Time-dependent Hamiltonian simulation}
\label{alg:timeDependentSimulation}
\small
\begin{algorithmic}[1]
\Require A Hamiltonian \(H\colon[0,T]\to\operatorname{Herm}(\Sys)\) satisfying
         \(\norm{H(t)}\leq\alpha\) and
         \(\norm{H(t)-H(s)}\leq\beta|t-s|\) for all \(s,t\in[0,T]\),
         \(\HAMT\) access to its
         time-indexed block-encodings, and a target error
         \(0<\eps\leq1/2\).
\Ensure A unitary \(U_{\mathrm{sim}}\) such that, if \(\mathsf W\)
        collects all work registers, then every unit system state
        \(\ket\psi\) satisfies
        \[
        \norm{
        U_{\mathrm{sim}}(\ket0_{\mathsf W}\ket\psi)
        -\ket0_{\mathsf W}U_H(T)\ket\psi}
        \leq\eps.
        \]
\State Choose the smallest integer \(q\geq\max\{1,\alpha T\}\) such
       that \(\sqrt{\alpha T}(12e\alpha T/q)^q\leq\eps/12\), and
       choose the smallest power of two \(J\) satisfying
       \[
       J\geq\max\left\{
       q,\frac{\beta T^2}{\eps},
       \frac{(\alpha T)^{3/2}}{\sqrt{3\eps}}
       \right\}.
       \]
\State Set \(t_j=jT/J\) and \(w=\alpha T/J\) for
       \(j=0,\ldots,J-1\).  Define the Cayley steps and their
       time-ordered product by
       \[
       U_j^C:=
       \left(I-\frac{iw}{2\alpha}H(t_j)\right)
       \left(I+\frac{iw}{2\alpha}H(t_j)\right)^{-1},
       \qquad
       U_C:=U_{J-1}^C\cdots U_0^C.
       \]
\State Construct the one-query Cayley transducer \(S\) and its
       catalyst map \(\Gamma\), characterized by
       \[
       S\bigl(\ket\psi\oplus\Gamma\ket\psi\bigr)
       =U_C\ket\psi\oplus\Gamma\ket\psi,
       \qquad \norm{\Gamma}^2\leq\alpha T.
       \]
\State For each reuse length \(N=1,\ldots,4q\), form the finite-reuse
       circuit from \(N\) calls to \(S\) with zero private input, and
       denote the encoded operator by \(P_N\).
\State Define \(G_N(z):=N^{-1}\sum_{k=0}^{N-1}z^k\) and
       \(Q(z):=(1+z^2)/2\), and choose the real coefficients
       \(\lambda_1,\ldots,\lambda_{4q}\) satisfying
       \[
       Q(z)^qG_{2q}(z)
       =\sum_{N=1}^{4q}\lambda_NG_N(z).
       \]
\State Apply the LCU construction to block-encode
       \(\widetilde U:=\sum_{N=1}^{4q}\lambda_NP_N\) with
       normalization \(2\), and apply one degree-three robust
       oblivious amplitude-amplification step to obtain
       \(U_{\mathrm{sim}}\).
\State \Return \(U_{\mathrm{sim}}\).
\end{algorithmic}
\end{algorithm}

\begin{proof}
If \(\alpha T\leq\eps\), then
\(\norm{U_H(T)-I}\leq\int_0^T\norm{H(t)}\,\mathrm{d}t\leq\alpha T\), so the
identity circuit proves the claim without an oracle query.  Hence assume
\(\alpha T>\eps\) and fix a unit initial state \(\ket{\psi_0}\).  Write
\(U_{\mathrm{sim}}^{(J)}\) for the amplified circuit constructed from
the \(J\)-step Cayley transducer.  Let
\(L=\log(12/\eps)\).  Since \(L\geq\log 24\) and
\(\alpha T>\eps\), we have
\(\log(e+L/(\alpha T))\leq2L\), so
\(L/\log(e+L/(\alpha T))\geq1/2\).  Choose
\[
 q=\left\lceil c_\star\left(
 \alpha T+\frac{L}{\log(e+L/(\alpha T))}
 \right)\right\rceil
\]
for a sufficiently large universal constant \(c_\star\).  This choice has
the order stated in the theorem.  We first show that
\(\sqrt{\alpha T}(12e\alpha T/q)^q\leq\eps/12\).
Suppose \(L\leq\alpha T\).  Increasing
\(c_\star\) if necessary gives
\[
 q\log\!\left(\frac{q}{12e\alpha T}\right)
 \geq c_\star\alpha T\log\!\left(\frac{c_\star}{12e}\right)
 \geq2\alpha T
 \geq L+\frac12\log(\alpha T).
\]
If \(L>\alpha T\), set \(x=L/(\alpha T)>1\).  Using
\(\log(e+x)\leq2\sqrt{x}\) and \(e+x\leq(e+1)x\), we obtain, for
sufficiently large \(c_\star\),
\[
 \log\!\left(\frac{q}{12e\alpha T}\right)
 \geq\log\!\left(\frac{c_\star\sqrt{x}}{24e}\right)
 \geq\frac14\log(e+x).
\]
Together with \(q\geq c_\star L/\log(e+x)\), this gives
\[
 q\log\!\left(\frac{q}{12e\alpha T}\right)
 \geq\frac{c_\star L}{4}
 \geq L+\frac12\max\{\log(\alpha T),0\}.
\]
The same final inequality holds in the first case.  Consequently,
\[
 \sqrt{\alpha T}
 \left(\frac{12e\alpha T}{q}\right)^q
 =\exp\!\left[
  \frac12\log(\alpha T)
  -q\log\!\left(\frac{q}{12e\alpha T}\right)
 \right]
 \leq e^{-L}=\frac{\eps}{12}.
\]

Choose \(J\) to be the smallest power of two satisfying
\begin{equation}
 J\geq\max\left\{
 q,\frac{\beta T^2}{\eps},
 \frac{(\alpha T)^{3/2}}{\sqrt{3\eps}}
 \right\}.
 \label{eq:finiteCayleyJ}
\end{equation}
Since \(J\geq q\),
\(w=\alpha T/J\leq\alpha T/q\).  Proposition~\ref{prop:clockFilter}
gives \(\eta:=\big\|U_C-\widetilde U\big\|\leq\eps/12<1/8\), so
Lemma~\ref{lem:robustOAA} applies with \(U=U_C\).  Together with
\eqref{eq:CayleyEvolutionError}, it gives
\begin{equation}
\begin{aligned}
 &\norm{
 U_{\mathrm{sim}}^{(J)}
 \bigl(\ket0_{\mathsf W}\ket{\psi_0}\bigr)
 -\ket0_{\mathsf W}U_H(T)\ket{\psi_0}}\leq
 3\sqrt{\alpha T}
 \left(\frac{12e\alpha T}{q}\right)^q
 +\frac{\beta T^2}{2J}
 +\frac{(\alpha T)^3}{12J^2}
 \leq\eps.
\end{aligned}
\label{eq:finiteCayleyTotalError}
\end{equation}
Indeed, the three terms in the last bound are at most
\(\eps/4\), \(\eps/2\), and \(\eps/4\), respectively, by the
choices of \(q\) and \(J\).

The query count is \eqref{eq:querycount}.  The amplification circuit
applies \(S^\circ\) or \((S^\circ)^\dagger\) \(12q\) times, and each
application contains \(J\) updates of the form \(R_j\).  By the gate
bound following \eqref{eq:globalUpdate}, these updates cost
\(\widetilde O(qJ(a+1))\) gates.  The remaining operations cost
\(\widetilde O(q+a)\): preparing the coefficient state over \(O(q)\)
labels costs \(O(q)\), the \(O(q)\) controlled state preparations,
shifts, and comparisons have polylogarithmic cost per operation, and
the two reflections cost \(\widetilde O(a)\).  Hence the complete gate count is
\[
 \widetilde O\bigl(qJ(a+1)+q+a\bigr)
 =\widetilde O\bigl(qJ(a+1)\bigr),
\]
where the equality uses \(J\geq q\geq1\).  The choice
\eqref{eq:finiteCayleyJ} satisfies
\[
 J=O\!\left(
 \max\left\{q,\frac{\beta T^2}{\eps},
 \frac{(\alpha T)^{3/2}}{\sqrt\eps}\right\}
 \right),
\]
which gives the gate bound in the theorem.  The tilde suppresses
polylogarithmic overhead from equality tests, reversible arithmetic,
controls on the time-index and reuse registers, and single-qubit
rotation synthesis.  For a fixed
universal gate set, run the construction with target error \(\eps/2\)
and synthesize the required rotations so that their accumulated error
is at most \(\eps/2\).  This changes only the suppressed polylogarithmic
factor and yields the stated error \(\eps\).
\end{proof}

\section*{Acknowledgements}

The authors thank Tongyang Li and Zhengfeng Ji for helpful discussions and valuable comments.

\paragraph{Statement on AI use.} The project was initiated and motivated by the authors. The main ideas underlying the work, including the central proof strategies, were generated by large language models. The human authors carefully studied, refined, and verified these ideas, and substantially rewrote and reorganized their presentation, adding the motivation and exposition necessary to communicate the arguments clearly. The final paper reflects the authors' own understanding of the results, and the authors take full responsibility for every claim, proof, and citation in it.

\begingroup
\small
\bibliographystyle{alphaurl}
\bibliography{ref,alpharef}
\endgroup

\appendix

\section{Construction with a continuous time register}
\label{app:continuousClock}
This appendix gives a continuous-time counterpart of the finite-step
transducer in the main text.  It can be viewed as the construction suggested
by taking the finite time-step size to zero.  We first fix the sampled
Hamiltonian that the continuous transducer will implement.

Fix a power of two \(M\), and define the sampling intervals and the
piecewise-constant Hamiltonian by
\[
 I_m:=\left[\frac{mT}{M},\frac{(m+1)T}{M}\right),\qquad
 H_M(t):=H(mT/M)\quad(t\in I_m).
\]
Let \(U_M(t,s)\) be the propagator generated by \(H_M\), and write
\(U_M:=U_M(T,0)\).  The oracle used in this appendix is
\[
 \HAMT_M=\sum_{m=0}^{M-1}\ketbra{m}{m}_{\mathsf T}
 \otimes O_{mT/M}.
\]
For \(t\in I_m\), set \(O_t:=O_{mT/M}\).  These operators satisfy
\begin{equation}
 O_t^\dagger=O_t,\qquad O_t^2=I,
 \label{eq:continuousOtProperties}
\end{equation}
and
\begin{equation}
 (\bra0_{\mathsf A}\otimes I)O_t
 (\ket0_{\mathsf A}\otimes I)=\frac{H_M(t)}{\alpha}.
 \label{eq:continuousOtBlock}
\end{equation}
The Lipschitz bound and Duhamel's formula give
\begin{equation}
 \norm{U_H(T)-U_M}
 \leq\int_0^T\norm{H(t)-H_M(t)}\,\mathrm{d}t
 \leq\frac{\beta T^2}{2M}.
 \label{eq:timeDiscretization}
\end{equation}


\subsection{An exact transducer on a continuous time register}

In this subsection, take the public space to be \(\Sys\) and the
private space to be
\[
 \Priv=L^2([0,T])\otimes\Anc\otimes\Sys
\]
and use generalized time kets \(\ket t_{\mathsf T}\).  On this space,
the oracle \(\HAMT_M\) is represented by the multiplication operator
\begin{equation}
 O=\int_0^T
 \ketbra{t}{t}_{\mathsf T}\otimes O_t\,\mathrm{d}t.
 \label{eq:continuousOracle}
\end{equation}
Equation~\eqref{eq:continuousOracle} is used only to analyze the given
oracle \(\HAMT_M\); it does not define a new access model.

Define the Volterra integral operator
\begin{equation}
 K_0=\int_0^T\int_0^t
 e^{-\alpha(t-s)}\ketbra{t}{s}_{\mathsf T}\,\mathrm{d}s\,\mathrm{d}t
 \label{eq:volterraKernel}
\end{equation}
on \(L^2([0,T])\).  As in Section~\ref{subsec:onequery}, let
\(\iota_0=\ket0_{\mathsf A}\otimes I\) and
\(\Pi_{\mathsf A,0}=\iota_0\iota_0^\dagger\), and set
\(\widehat\Pi_0=I_{L^2([0,T])}\otimes\Pi_{\mathsf A,0}\).
The kernel in \eqref{eq:volterraKernel} vanishes when the output time
precedes the input time.  The factor \(e^{-\alpha(t-s)}\) is part of the
unitary construction below and does not describe dissipative dynamics.

The following four blocks are independent of the oracle and
act between the public and private spaces:
\begin{align}
 A_0&=iI-2i\alpha\widehat\Pi_0
 (K_0\otimes I_{\mathsf{AS}})\widehat\Pi_0,\nonumber\\
 B&=\sqrt{2\alpha}\int_0^T e^{-\alpha t}
 \ket t_{\mathsf T}\ket0_{\mathsf A}\otimes I_{\mathsf S}\,\mathrm{d}t,
 \nonumber\\
 C_0&=-i\sqrt{2\alpha}\int_0^T e^{-\alpha(T-s)}
 \bra s_{\mathsf T}\bra0_{\mathsf A}\otimes I_{\mathsf S}\,\mathrm{d}s,
 \nonumber\\
 D&=e^{-\alpha T}I_{\mathsf S}.
 \label{eq:continuousColligationBlocks}
\end{align}
Write
\begin{equation}
 R=\begin{pmatrix}D&C_0\\ B&A_0\end{pmatrix}
 \colon\Sys\oplus\Priv\longrightarrow\Sys\oplus\Priv.
 \label{eq:continuousColligation}
\end{equation}

\begin{lemma}[Unitarity of \(R\)]
\label{lem:continuousColligation}
The operator \(R\) is unitary and is independent of the Hamiltonian
oracle.
\end{lemma}

\begin{proof}
Suppose
\(R(\ket\psi\oplus\ket f)=\ket{q_T}\oplus\ket g\).
Represent the private input and output by the functions
\(t\mapsto\ket{f_t}\) and \(t\mapsto\ket{g_t}\), and put
\(\ket{f_{0,t}}=\iota_0^\dagger\ket{f_t}\).  The definitions above are
equivalent to
\begin{align}
 \frac{\mathrm{d}}{\mathrm{d}t}\ket{q_t}
 &=-\alpha\ket{q_t}-i\sqrt{2\alpha}\ket{f_{0,t}},
 \qquad \ket{q_0}=\ket\psi,\nonumber\\
 \ket{g_t}
 &=i\ket{f_t}
   +\sqrt{2\alpha}\,\iota_0\ket{q_t}
 \label{eq:continuousInputOutput}
\end{align}
Differentiating \(\norm{\ket{q_t}}^2\)
gives
\[
 \frac{\mathrm{d}}{\mathrm{d}t}\norm{\ket{q_t}}^2
 =\norm{\ket{f_{0,t}}}^2-\norm{\iota_0^\dagger\ket{g_t}}^2.
\]
The second line of \eqref{eq:continuousInputOutput} gives
\(
 \norm{(I-\Pi_{\mathsf A,0})\ket{g_t}}
 =\norm{(I-\Pi_{\mathsf A,0})\ket{f_t}}
\)
for almost every \(t\).  Integrating this equality together with the
preceding differential identity gives
\[
 \norm{\ket g}^2+\norm{\ket{q_T}}^2
 =\norm{\ket f}^2+\norm{\ket\psi}^2.
\]
Thus \(R\) is an isometry.  To prove that it is onto, prescribe an
arbitrary output \(\ket{q_T}\oplus\ket g\), solve
\(\dot q_t=\alpha q_t-\sqrt{2\alpha}\,
\iota_0^\dagger g_t\) backwards from \(T\), and recover
\[
 \ket{f_{0,t}}
 =i\bigl(\sqrt{2\alpha}\ket{q_t}
 -\iota_0^\dagger\ket{g_t}\bigr),\qquad
 (I-\Pi_{\mathsf A,0})\ket{f_t}
 =-i(I-\Pi_{\mathsf A,0})\ket{g_t}.
\]
This constructs an inverse for every output and proves unitarity.
\end{proof}

Apply the oracle once on the private input before applying \(R\):
\begin{equation}
 S=R(I\oplus O)
 =\begin{pmatrix}D&C\\ B&A\end{pmatrix},
 \qquad A=A_0O,\quad C=C_0O.
 \label{eq:continuousTransducerBlocks}
\end{equation}
For the propagator \(U_M\), define
\begin{equation}
 \Phi=\sqrt\alpha\int_0^T
 \ket t_{\mathsf T}\ket0_{\mathsf A}\otimes U_M(t,0)\,\mathrm{d}t,
 \qquad
 \Gamma=\frac{I+iO}{\sqrt2}\Phi.
 \label{eq:continuousCatalyst}
\end{equation}

\begin{proposition}[Exact transducer for \(U_M\)]
\label{prop:continuousTransducer}
The unitary \(S\) uses one query to \(\HAMT_M\) and, for every
initial state \(\ket{\psi_0}\in\Sys\), satisfies
\begin{equation}
 S\bigl(\ket{\psi_0}\oplus\Gamma\ket{\psi_0}\bigr)
 =U_M\ket{\psi_0}\oplus\Gamma\ket{\psi_0},
 \qquad
 \Gamma^\dagger\Gamma=\alpha T I.
 \label{eq:continuousGraph}
\end{equation}
\end{proposition}

\begin{proof}
Fix \(\ket{\psi_0}\), write
\(\ket{v_t}=U_M(t,0)\ket{\psi_0}\), and take
\(\ket f=O\Gamma\ket{\psi_0}\) as the private input to \(R\).
Projecting \(\ket f_t\) onto \(\ket0_{\mathsf A}\) and using
\eqref{eq:continuousCatalyst} and \eqref{eq:continuousOtBlock} gives
\[
 \ket{f_{0,t}}
 =\frac{H_M(t)\ket{v_t}}{\sqrt{2\alpha}}
  +i\sqrt{\frac\alpha2}\ket{v_t}.
\]
Substitution into the first line of
\eqref{eq:continuousInputOutput} cancels the damping term and gives
\(\dot v_t=-iH_M(t)v_t\).  Substituting
\(\ket{q_t}=\ket{v_t}\) and
\(\ket{f_t}=(O\Gamma\ket{\psi_0})_t\) into the last two lines gives
\(\ket{g_t}=(\Gamma\ket{\psi_0})_t\) for almost every \(t\).
Thus the private output is the original catalyst and the public output
is \(\ket{v_T}=U_M\ket{\psi_0}\), proving the transducer identity.

For the catalyst norm, \((I+iO)/\sqrt2\) is unitary because \(O\) is Hermitian and
satisfies \(O^2=I\), pointwise by
\eqref{eq:continuousOtProperties}.  Using the
definition of \(\Phi\) in \eqref{eq:continuousCatalyst} and the
unitarity of \(U_M(t,0)\),
\[
 \Phi^\dagger\Phi
 =\alpha\int_0^T U_M(t,0)^\dagger U_M(t,0)\,\mathrm{d}t
 =\alpha T I,
\]
which proves the catalyst identity.
\end{proof}

\subsection{Factorial bound and combination of reuse lengths}

We next prove the factorial bound directly for the private-to-private
block of this transducer.  Set
\[
 L_0=\alpha\widehat\Pi_0
 (K_0\otimes I_{\mathsf{AS}})\widehat\Pi_0.
\]
Equation~\eqref{eq:continuousTransducerBlocks}, together with
\(O^2=I\) from \eqref{eq:continuousOtProperties}, gives
\begin{equation}
 A=i(I-2L_0)O,\qquad
 I+A^2=2L_0+2OL_0O-4L_0OL_0O.
 \label{eq:continuousPairedExpansion}
\end{equation}
Substituting the private-to-private block \(A\) into
\(Q(z)=(1+z^2)/2\) from \eqref{eq:Qpoly} gives
\begin{equation}
 Q:=\frac{I+A^2}{2}
 =L_0+OL_0O-2L_0OL_0O.
 \label{eq:continuousQ}
\end{equation}
Because the identity term has canceled and every remaining term
contains \(L_0\), \(Q\) is causal: its integral kernel is supported
only where the input time precedes the output time.

For a private vector \(\ket f\), represent it by the function
\(s\mapsto\ket{f_s}\).  Define the operator-valued kernel \(Q(t,s)\) by
\[
 (\bra t_{\mathsf T}\otimes I)Q\ket f
 =\int_0^TQ(t,s)\ket{f_s}\,\mathrm{d}s.
\]
The kernel vanishes for \(s\geq t\), apart from the measure-zero
diagonal.  Suppose \(s<t\), with \(t\) in sampling interval \(m\) and
\(s\) in sampling interval \(\ell\).  Using \eqref{eq:continuousQ} and the kernel
in \eqref{eq:volterraKernel}, the first two terms have kernels
\[
 \alpha e^{-\alpha(t-s)}\Pi_{\mathsf A,0},
 \qquad
 \alpha e^{-\alpha(t-s)}
 O_t\Pi_{\mathsf A,0}O_s.
\]
Every oracle block is unitary and
\(\Pi_{\mathsf A,0}\) is an orthogonal projection.  The remaining
term therefore satisfies
\[
 \norm{(L_0OL_0O)(t,s)}
 \leq\alpha^2(t-s)e^{-\alpha(t-s)}.
\]
The triangle inequality gives
\begin{equation}
 \norm{Q(t,s)}
 \leq\kappa(t-s),\qquad
 \kappa(\tau):=2\alpha e^{-\alpha\tau}(1+\alpha\tau),
 \quad \tau\geq0.
 \label{eq:continuousKernelBound}
\end{equation}

For \(q\geq1\), define the causal convolution powers by
\[
 \kappa^{*1}=\kappa,\qquad
 \kappa^{*(q+1)}(\tau)
 =\int_0^\tau\kappa(\tau-u)\kappa^{*q}(u)\,\mathrm{d}u.
\]
The kernel of \(Q^q\) integrates only over intermediate times
\(s<u_1<\cdots<u_{q-1}<t\), the continuous analogue of the
nondecreasing index order in the discrete proof.  Composition of causal
integral kernels and
\eqref{eq:continuousKernelBound} therefore give inductively, for
\(0\leq s<t\leq T\),
\begin{equation}
 \norm{Q^q(t,s)}\leq\kappa^{*q}(t-s)
 \label{eq:continuousConvolutionBound}
\end{equation}

Schur's test and causality imply
\begin{equation}
 \norm{Q^q}
 \leq\int_0^T\kappa^{*q}(\tau)\,\mathrm{d}\tau.
 \label{eq:continuousSchur}
\end{equation}
For any \(\theta>0\), the definition of \(\kappa\) in
\eqref{eq:continuousKernelBound}, its nonnegativity, and the
Laplace-transform rule for convolution give
\[
\begin{aligned}
 \int_0^T\kappa^{*q}(\tau)\,\mathrm{d}\tau
 &\leq e^{\theta T}
 \left(\int_0^\infty e^{-\theta\tau}\kappa(\tau)\,\mathrm{d}\tau\right)^q\\
 &=e^{\theta T}
 \left(
 \frac{2\alpha}{\theta+\alpha}
 +
 \frac{2\alpha^2}{(\theta+\alpha)^2}
 \right)^q.
\end{aligned}
\]
Taking \(\theta=q/T\) and using \(q\geq\alpha T\), the quantity in
parentheses is at most \(4\alpha T/q\).  We have proved the following
continuous analogue of Proposition~\ref{prop:factorialDecay}.

\begin{proposition}[Factorial bound for the continuous transducer]
\label{prop:continuousFactorial}
For every integer \(q\geq\max\{1,\alpha T\}\),
\begin{equation}
 \norm{Q^q}
 \leq\left(\frac{4e\alpha T}{q}\right)^q.
 \label{eq:continuousFactorial}
\end{equation}
\end{proposition}

For the continuous transducer, define \(P_N\) by
\eqref{eq:VNpublicOutput}.  Since the proof of
Lemma~\ref{lem:zeroCatalystBlock} uses only the block identities of a
unitary transducer,
\[
 U_M-P_N=C\,G_N(A)\Gamma.
\]
Use the coefficients \(\lambda_N\) obtained from \(E_q\) through
\eqref{eq:clockExpansion} and set
\(\widetilde U=\sum_{N=1}^{4q}\lambda_NP_N\).  Summing the preceding
equation with weights \(\lambda_N\), using \(\sum_N\lambda_N=1\), and
substituting \(E_q(z)=Q(z)^qG_{2q}(z)\) from \eqref{eq:Eq} gives
\[
 U_M-\widetilde U
 =C\,E_q(A)\Gamma
 =C\,Q^qG_{2q}(A)\Gamma.
\]
Since \(A\) and \(C\) are contractions,
\(\norm{G_{2q}(A)}\leq1\).  Proposition
\ref{prop:continuousTransducer} and
Proposition~\ref{prop:continuousFactorial} therefore give
\begin{equation}
 \norm{U_M-\widetilde U}
 \leq\sqrt{\alpha T}
 \left(\frac{4e\alpha T}{q}\right)^q.
 \label{eq:continuousFilterError}
\end{equation}
Proposition~\ref{prop:clockFilter} gives
\(\sum_N|\lambda_N|<2\), and the largest reuse length remains \(4q\).

\subsection{Finite-dimensional restriction}
\label{subsec:finiteRestriction}

We now identify a finite-dimensional subspace containing every private
state reachable from zero private input in at most \(4q\) calls.  For each
sampling interval \(I_m=[mT/M,(m+1)T/M)\) and integer \(a\geq0\),
define
\begin{equation}
 \ket{m,a}_{\mathsf T}
 :=\int_{mT/M}^{(m+1)T/M}
 e^{-\alpha t}(t-mT/M)^a\ket t_{\mathsf T}\,\mathrm{d}t.
 \label{eq:finiteClockKets}
\end{equation}
Here \(m\) labels a sampling interval; there is no Cayley-step index in
this construction.  These vectors specify a subspace, not an orthonormal
computational basis.  Let
\begin{equation}
 \mathcal V_k=\operatorname{span}
 \{\ket{m,a}_{\mathsf T}:0\leq m<M,\ 0\leq a\leq k\}.
 \label{eq:finiteClockSpan}
\end{equation}
The supports of different sampling intervals are disjoint, and the
polynomials within each interval are linearly independent.  Hence
\(\dim\mathcal V_k=M(k+1)\).

The multiplication oracle preserves every \(\mathcal V_k\), because
\(O_t\) is constant on each sampling interval.  The operator \(K_0\)
defined in \eqref{eq:volterraKernel} increases the polynomial degree
within an interval by at most one:
\begin{equation}
 K_0\ket{m,a}_{\mathsf T}
 =\frac{\ket{m,a+1}_{\mathsf T}}{a+1}
 +\frac{(T/M)^{a+1}}{a+1}
  \sum_{\ell=m+1}^{M-1}\ket{\ell,0}_{\mathsf T}.
 \label{eq:finiteClockRecurrence}
\end{equation}
The public-to-private block \(B\) from
\eqref{eq:continuousColligationBlocks} satisfies
\[
 B\ket\psi
 =\sqrt{2\alpha}\sum_{m=0}^{M-1}
 \ket{m,0}_{\mathsf T}\ket0_{\mathsf A}\ket\psi
 \in\mathcal V_0\otimes\Anc\otimes\Sys.
\]
Since each subsequent application of \(A=A_0O\) from
\eqref{eq:continuousTransducerBlocks} raises the degree by at most one,
after \(k\) transducer calls every private state
reachable from zero private input lies in
\(\mathcal V_k\otimes\Anc\otimes\Sys\).

Choose an orthonormal basis separately within the subspace for each
sampling interval.
In this basis the oracle is block diagonal in the sampling-interval
index \(m\) and acts as the identity on the basis index within each
interval.
It is therefore \(\HAMT_M\) tensored with an identity, and still counts
as one query.
At call \(k\), the operator \(R\) in
\eqref{eq:continuousColligation} maps
\[
 \Sys\oplus(\mathcal V_k\otimes\Anc\otimes\Sys)
 \quad\hbox{isometrically into}\quad
 \Sys\oplus(\mathcal V_{k+1}\otimes\Anc\otimes\Sys).
\]
Extend this isometry to a unitary on
\(\Sys\oplus(\mathcal V_{4q}\otimes\Anc\otimes\Sys)\), and use its
adjoint in the inverse circuit.  Since the private state before call
\(k\) lies in \(\mathcal V_k\otimes\Anc\otimes\Sys\), induction shows
that these finite-dimensional unitaries agree with the continuous
construction for every sequence of at most \(4q\) transducer calls.
They therefore reproduce every reuse circuit entering
\(\widetilde U\) in \eqref{eq:continuousFilterError}.

The relevant subspace has dimension
\begin{equation}
 \dim\mathcal V_{4q}=M(4q+1).
 \label{eq:finiteClockDimension}
\end{equation}
The restriction introduces no truncation error and preserves the
identity \([V]_0=\widetilde U/2\) from \eqref{eq:half}.  Choosing \(M\)
to control
\eqref{eq:timeDiscretization}, and then choosing \(q\) as in the proof
of Theorem~\ref{thm:main}, the LCU and amplification construction in
Section~\ref{sec:combining} gives the same asymptotic query scaling,
with \(12q\) queries to \(\HAMT_M\).

We do not give an efficient circuit for \(K_0\) restricted to
\(\mathcal V_{4q}\) or for a unitary extension of \(R\) on the
finite-dimensional subspaces used above.  The unitary completion is
only an existence argument.  We
therefore make no gate-complexity claim for this construction, and it
does not improve the one- and two-qubit gate bound in
Theorem~\ref{thm:main}.

\end{document}